\documentclass{amsart}
\usepackage{amsmath,amsfonts,amsthm}

\usepackage{cite}
\usepackage{color,xcolor}
\usepackage{comment}
\usepackage[normalem]{ulem}%for strikeout
\usepackage{array}%extended implementation of array and tabular
	\newcolumntype{L}{>{\displaystyle}l}
	\newcolumntype{C}{>{\displaystyle}c}
	\newcolumntype{R}{>{\displaystyle}r}
\usepackage[linktocpage]{hyperref}
\newcommand{\reff}[1]{\hyperref[#1]{\ref{#1}}}
\definecolor{dark-blue}{rgb}{0,0,.6}
\hypersetup{
  colorlinks= true,	%Colours links instead of ugly boxes
  urlcolor= blue,		%Colour for external hyperlinks
  linkcolor= blue,	%Colour of internal links
  citecolor= blue,	%Colour of citations
  linkbordercolor=blue,
  citebordercolor=blue
}

\newcommand{\Z}{\ensuremath{\mathbb{Z}}}
\newcommand{\R}{\ensuremath{\mathbb{R}}}

\newcommand{\N}{\ensuremath{\mathbb{N}}}

\newcommand{\B}{\ensuremath{\mathbb{B}}}
\newcommand{\T}{\ensuremath{\mathbb{T}}}

\renewcommand{\S}{\ensuremath{\mathbb{S}}}

\newcommand{\calA}{\ensuremath{\mathcal{A}}}

\newcommand{\calL}{\ensuremath{\mathcal{L}}}
\newcommand{\calM}{\ensuremath{\mathcal{M}}}

\newcommand{\calU}{\ensuremath{\mathcal{U}}}
\newcommand{\calS}{\ensuremath{\mathcal{S}}}

\newcommand{\bfa}{\ensuremath{\mathbf{a}}}

\newcommand{\bfzero}{\ensuremath{\mathbf{0}}}
\newcommand{\bfone}{\ensuremath{\mathbf{1}}}

\newcommand{\bs}{\backslash}
\newcommand{\inv}{\ensuremath{^{-1}}}

\newcommand{\defword}[1]{{\it #1}}

\renewcommand{\epsilon}{\varepsilon}
\renewcommand{\phi}{\varphi}
\renewcommand{\th}{^{\textnormal{th}}}

\DeclareMathOperator{\Fix}{Fix}

\DeclareMathOperator{\tr}{tr}
\DeclareMathOperator{\Aut}{Aut}
\DeclareMathOperator{\End}{End}

\DeclareMathOperator{\GL}{GL}
\DeclareMathOperator{\SL}{SL}
\DeclareMathOperator{\supp}{supp}

\newcommand{\set}[1]{\ensuremath{\left\{ #1 \right\}}}
\newcommand{\brac}[1]{\ensuremath{\langle #1 \rangle}}

\newtheorem{theorem}{Theorem}[section]
\newtheorem{lemma}[theorem]{Lemma}
\newtheorem{proposition}[theorem]{Proposition}
\newtheorem{corollary}[theorem]{Corollary}

\newtheorem*{theorem*}{Theorem}
\newtheorem*{thm:linear}{Theorem \reff{T:LinearGapIntro}}
\newtheorem*{thm:hausdim}{Theorem \reff{T:HausDimIntro}}
\newtheorem*{thm:gaplabel}{Theorem \reff{T:GapLabelingIntro}}
\newtheorem*{thm:exactdim}{Theorem \reff{T:ExactDimIntro}}
\newtheorem*{thm:lessthandim}{Theorem \reff{T:LessThanDimIntro}}

\theoremstyle{definition}
\newtheorem{definition}[theorem]{Definition}

\theoremstyle{remark}
\newtheorem{remark}[theorem]{Remark}

\title[Primitive Invertible Substitution Schr\"odinger Operators]{Spectra of Discrete Schr\"odinger Operators with Primitive Invertible Substitution Potentials}
\author{May Mei}
\address{Mathematics \& Computer Science, Denison University, Granville, OH 43023-0810}
\email{meim@denison.edu}
\thanks{The author gratefully acknowledges the support of the Michele T. Myers PD Account through Denison University. The author was also partially supported by the NSF grants DMS-1301515 (PI: A. Gorodetski) and IIS-1018433 (PI: M. Welling, Co-PI: A. Gorodetski).}
\date{\today}
\begin{document}

\begin{abstract}
We study the spectral properties of discrete Schr\"odinger operators with potentials given by primitive invertible substitution sequences (or by Sturmian sequences whose rotation angle has an eventually periodic continued fraction expansion, a strictly larger class than primitive invertible substitution sequences). It is known that operators from this family have spectra which are Cantor sets of zero Lebesgue measure. We show that the Hausdorff dimension of this set tends to $1$ as coupling constant $\lambda$ tends to $0$. Moreover, we also show that at small coupling constant, all gaps allowed by the gap labeling theorem are open and furthermore open linearly with respect to $\lambda$. Additionally, we show that, in the small coupling regime, the density of states measure for an operator in this family is exact dimensional. The dimension of the density of states measure is strictly smaller than the Hausdorff dimension of the spectrum and tends to $1$ as $\lambda$ tends to $0$.
\end{abstract}

\maketitle

\tableofcontents

%%%%%%%%%%INTRODUCTION
%%%%%%%%%%INTRODUCTION
%%%%%%%%%%INTRODUCTION
%%%%%%%%%%INTRODUCTION
%%%%%%%%%%INTRODUCTION

\section{Introduction}\label{S:Introduction}

Let $\lambda>0$, let $\alpha \in \T^1:=\R/\Z$, be irrational, and let $\omega \in \T^1$. The discrete Schr\"odinger operator $H_{\lambda, \alpha, \omega}$ is a bounded, self-adjoint operator on $\ell^2(\Z)$ given by
\begin{equation}\label{E:SchrodingerOperator}
(H_{\lambda, \alpha, \omega} \psi)(n)=\psi(n+1)+\psi(n-1)+\lambda \cdot R_{\alpha, \omega}(n) \cdot\psi(n),\quad \psi \in \ell^2(\Z),
\end{equation}
where
\begin{equation}\label{E:RotationSequence}
R_{\alpha, \omega}(n):=\chi_{[1-\alpha,1)}(n\alpha + \omega \mod 1)
\end{equation}
 is a \defword{rotation sequence} of \defword{angle} $\alpha$ and \defword{initial point} or \defword{phase} $\omega$. In general, this sequence in the third summand is referred to as the \defword{potential} and $\lambda$ is called the \defword{coupling constant}. 

This operator with this particular choice of potential is a popular model in the study of electronic properties of quasicrystals, whose discovery by Dan Shechtman earned him the 2011 Nobel Prize in Chemistry \cite{SBGC84}. Rotation sequences are \defword{Sturmian}, i.e. non-periodic sequences of the lowest possible complexity (a formal definition is given in Section~\reff{S:Sturmian}). We are motivated in the choice of this class of potentials by the heuristic that crystals are perfectly ordered (periodic) and quasicrystals are aperiodic, but still highly ordered. See \cite{DEG13} for a survey of results regarding Schr\"odinger operators (including higher dimensions) that model quasicrystals. Other related survey papers include \cite{Bel92b}, \cite{BG95}, \cite{Dam07}, and \cite{Sut95}.

Primitive invertible substitutions, which are rigorously defined in Section~\reff{S:Substitution}, are another method of generating Sturmian sequences, although not all Sturmian sequences can be achieved as a substitution sequence. Refer to the survey papers in the previous paragraph, and to \cite{Can09}.  After this paper was prepared, we learned that some related results were obtained by Arnaud Girand in \cite{Gir13}. The canonical example of a primitive invertible substitution is the Fibonacci substitution $s_F: 0 \mapsto 01, 1 \mapsto 0$. This substitution generates the sequence $R_{\frac{\sqrt{5}-1}{2}, 0}$ (details are provided in Section~\reff{S:Sturmian}). The spectral properties of the corresponding operator, the Fibonacci Hamiltonian, has been thoroughly studied in \cite{Cas86}, \cite{Sut87}, and \cite{Sut89}, and more recently in \cite{DEGT08}, \cite{DG09a}, \cite{DG09b}, \cite{DG11}, and \cite{DG12}.

%%%%%

In this paper, we generalize some results previously known for the Fibonacci Hamiltonian to any discrete Schr\"odinger operator with primitive invertible substitution sequence as its potential. Forthcoming joint work with with William Yessen \cite{MY14} extends these results to Jacobi operators. A standard argument, which uses the minimality of the left shift on the dynamical hull generated by the substitution sequence (see Section~\reff{S:Sturmian} of this paper and \cite{Dam07}) shows that the spectrum of $H_{\lambda, \alpha, \omega}$ is independent of initial point $\omega$. Thus, we drop $\omega$ from the notation and denote the spectrum of the operator $H_{\lambda, \alpha}$ by $\Sigma_{\lambda, \alpha}$.

It was shown in \cite{Sut89} that the Fibonacci Hamiltonian has measure zero Cantor spectrum. This results was then extended to all Sturmian potentials in \cite{BIST89}. Later, \cite{DL06}, established Cantor spectrum of zero Lebesgue measure for a large class of low-complexity potentials. See also \cite{Len01} and \cite{LTWW02}. It is natural to ask about the behavior of gaps in the Cantor spectrum as one varies the coupling constant $\lambda$.

\begin{theorem}\label{T:LinearGapIntro}
For $\alpha$ with eventually periodic continued fraction expansion and $\lambda>0$ sufficiently small, the boundary points of a gap in the spectrum $\Sigma_{\lambda,\alpha}$ depend smoothly on the coupling constant $\lambda$. Moreover, given any one-parameter continuous family $\set{U_{\lambda,\alpha}}_{\lambda>0}$ of gaps of $\Sigma_{\lambda,\alpha}$, we have that
$$\lim_{\lambda \to 0^+} \frac{|U_{\lambda,\alpha}|}{|\lambda|}$$
exists and belongs to $(0,\infty)$.
\end{theorem}

The Cantor structure of the spectrum also motivates the study of fractal properties of the spectrum, namely Hausdorff dimension, denoted $\dim$ as there is no ambiguity in this paper, and \defword{thickness}, which is defined in Section~\reff{S:ContHausDim}. As a corollary to \cite{Can09}, $\Sigma_{\lambda, \alpha}$ is a \defword{dynamically defined} Cantor set for $\alpha$ with eventually periodic continued fraction expansion and $\lambda>0$ sufficiently small. From this, we can immediately see that the local Hausdorff dimension of the spectrum $\Sigma_{\lambda, \alpha}$ is equal to the global Hausdorff dimension and this is also equal to the box counting dimension. It was already known from \cite{Can09} that $\dim(\Sigma_{\lambda,\alpha})$ is strictly between 0 and 1 for $\lambda>0$. Furthermore, the Hausdorff dimension of the spectrum $\Sigma_{\lambda, \alpha}$  depends continuously on $\lambda$. Refer to Section 2 of \cite{DG09a} and the references contained therein or to \cite{PT93} for general background on relevant theorems from dynamical systems.

\begin{theorem}\label{T:HausDimIntro}
For $\alpha$ with eventually periodic continued fraction expansion,
$$\lim_{\lambda \to 0^+} \dim \Sigma_{\lambda, \alpha} = 1.$$
More precisely, there are constants $C_1, C_2 > 0$ such that
$$1-C_1 \lambda \leq \dim \Sigma_{\lambda, \alpha} \leq 1-C_2 \lambda$$
for $\lambda>0$ sufficiently small.
\end{theorem}

Let $H^\mathfrak{I}_{\lambda, \alpha}$ be $H_{\lambda, \alpha}$ restricted to a finite interval $\mathfrak{I}$ with Dirichlet boundary conditions and let the \defword{integrated density of states} be given by
\begin{equation}\label{E:IDOS}
N_{\lambda, \alpha}(E)=\lim_{L \to \infty} \frac{1}{L} \left|\left\{\text{eigenvalues of }H_{\lambda, \alpha}^{[1,L]} \text{ that lie in }(-\infty,E)\right\}\right|.
\end{equation}
The existence of this limit has been shown in a general setting in \cite{LS05}. It has also been studied for general potentials, random potentials, and analytic quasi-periodic potentials, see Section 5.5 of \cite{DEG13} for references. This will be developed more formally in Section~\reff{S:IDOS}, but for now we simply state that $N_{\lambda, \alpha}$ is the cumulative distribution function of the averaged in $\omega$ spectral measure supported on $\Sigma_{\lambda, \alpha}$, which we denote $dN_{\lambda, \alpha}$ and call the \defword{density of states measure}. 

\begin{theorem}\label{T:ExactDimIntro}
For $\alpha$ with eventually periodic continued fraction expansion, there exists $0<\lambda_0\leq \infty$ such that for $\lambda \in (0, \lambda_0)$, there is $d_{\lambda, \alpha} \in (0,1)$ so that $dN_{\lambda, \alpha}$ is of exact dimension $d_{\lambda, \alpha}$, that is, for $dN_{\lambda, \alpha}$-almost every $E \in \R$,
$$\lim_{\epsilon \to 0^+} \frac{\log(N_{\lambda, \alpha}(E+\epsilon)-N_{\lambda,\alpha}(E-\epsilon))}{\log \epsilon}=d_{\lambda, \alpha}.$$
Moreover, in $(0, \lambda_0)$, $d_{\lambda, \alpha}$ is a $C^\omega$ function of $\lambda$ and $$\lim_{\lambda \to 0^+} d_{\lambda, \alpha}=1.$$
\end{theorem}

\begin{theorem}\label{T:LessThanDimIntro}
For $\alpha$ with eventually periodic continued fraction expansion and for $\lambda>0$ sufficiently small, we have $$d_{\lambda, \alpha}<\dim \Sigma_{\lambda, \alpha}.$$
\end{theorem}

The gap labeling theorem \cite{BBG92} provides a set of canonical labels for the cumulative distribution function of the spectral measure which correspond to gaps in the spectrum. See also \cite{Bel92a} for an extended summary of results regarding these labelings in a more general setting. It is of interest to know whether all possible gaps are open as the density of the labels indicates a Cantor spectrum.

\begin{theorem}\label{T:GapLabelingIntro}
For $\alpha$ with eventually periodic continued fraction expansion, there is $\lambda_0>0$ such that for every $\lambda \in (0, \lambda_0]$, all gaps of $H_{\lambda, \alpha}$ allowed by the gap labeling theorem are open. %That is $$\set{N(E, \lambda, \alpha): E \in \R \bs \Sigma_\lambda}=\set{{m\alpha} : m \in \Z} \cup \set{1}.$$
\end{theorem}

The famous ``Ten Martini Problem'' is showing that the spectrum of the almost Mathieu operator is a Cantor set. This was done by Avila and Jitomirskaya in 2009 \cite{AJ09}. See also \cite{Dam08} for a summary of previous work leading up to this result. A related question is the so-called ``Dry Ten Martini Problem,'' showing that all gaps are open.  In the above theorem, we solve the ``dry'' version of this problem for primitive invertible substitution potentials.

At this point, one may ask whether the results in this paper can be extended either to all Sturmian potentials or outside of a small coupling constant regime. To answer the former question, in \cite{LPW07}, it was shown that there exists an increasing sequence of integers $\set{m_k}_{k \in \N}$ such that for rotation angle $\tilde{\alpha}=[0;a_1, a_2,...]$ with $a_n=1$ if $m_{2k-1} \leq n < m_{2k}$ and $a_n=3$ if $m_{2k} \leq n < m_{2k+1}$ and for sufficiently large coupling constant $\lambda$, the Hausdorff dimension and box counting dimension of $\Sigma_{\lambda, \tilde{\alpha}}$ is not equal, implying that $\Sigma_{\lambda, \tilde{\alpha}}$ is not a dynamically defined Cantor set. Thus the results of this paper certainly do not hold for all Sturmian potentials. As to the latter question, the value of $\lambda_0$ in Theorems, for example ~\reff{T:ExactDimIntro}, corresponds to the largest value for which transversailty of the line of initial conditions can be shown. It may be the case that $\lambda_0=+\infty$, but we were unable to show this.

%%%%%

Now let us outline the structure of the paper. Section~\reff{S:Background} provides background on Sturmian sequences and substitution sequences. See \cite{Fog02} for a more thorough introduction to Sturmian sequences, substitution sequences, and related dynamical properties. It has long been known that any Sturmian sequence can be achieved as a rotation sequence of some angle and some initial point \cite{MH40}. It is also known that the rotation sequences with initial point 0 that are fixed by a nontrivial substitution are exactly those whose rotation angle $\alpha$ has a particular periodic continued fraction expansion \cite{CMPS93}, restated in Theorem~\reff{T:SturmianInvariant}. In fact, the results in this paper hold for all Sturmian sequences whose rotation angle has \emph{any} eventually periodic continued fraction expansion, a strictly larger class than primitive invertible substitution sequences.

Section~\reff{S:TraceMaps} introduces \defword{trace maps}, which are dynamical systems that appears in a natural way in the study of spectral properties of discrete Schr\"odniger operators with substitution potential. This was initially introduced by \cite{KKT83}, see also \cite{OPRSS83}, for the case of the Fibonacci Hamiltonian.

Section~\reff{S:HyperbolicityNW} establishes the hyperbolicity of the trace map and we get as a corollary to \cite{Can09} that that $\Sigma_{\lambda, \alpha}$ is a \defword{dynamically defined} Cantor set for $\lambda>0$ sufficiently small. The notion of a dynamically defined Cantor set is explored in \cite{PT93}.

Section~\reff{S:PeriodicCurve} demonstrates the existence of a curve of periodic points. In the Fibonacci case, this theorem was established through a computation that relied on the explicit form of the Fibonacci trace map. To generalize this result, we approach this problem for the class of operators with potentials given by rotation angles with eventually periodic continued fraction expansions by considering an application of bilinear forms. This is needed for the proof of Theorems~\reff{T:LinearGapIntro} and \reff{T:HausDimIntro}.

In Section~\reff{S:LinearGapOpening}, we prove Theorem~\reff{T:LinearGapIntro}, that is the gaps of the Cantor spectrum open linearly with respect to coupling constant. This is the analogue of Theorem 1.3 from \cite{DG11} for the Fibonacci Hamiltonian. This proof heavily utilizes the construction from Section~\reff{S:PeriodicCurve}.

Section~\reff{S:ContHausDim} is dedicated to the proof of Theorem ~\reff{T:HausDimIntro}, that the Hausdorff dimension of the spectrum is continuous for $\lambda=0$, which heavily utilizes thickness and the related \defword{denseness}. This is the analogue of Theorem 1.1 from \cite{DG11}. Quantitative properties of Cantor sets such as thickness and denseness have been widely used in in the setting of dynamical systems theory (see, for example, \cite{New70}, \cite{New79}, and \cite{dM73}). This result presents a nice application of thickness to a problem arising in mathematical physics.

In Section~\reff{S:IDOS}, we provide the proofs of Theorem~\reff{T:ExactDimIntro} and Theorem~\reff{T:LessThanDimIntro}, which are generalizations of results obtained for the Fibonacci Hamiltonian in \cite{DG12}. We also prove Theorem~\reff{T:GapLabelingIntro}, that all gaps allowed by the gap labeling theorem open. This is the analogue of Theorem 1.5 from \cite{DG11}.

\emph{Acknowledgements}: The author would like to acknowledge the invaluable contributions of Anton Gorodetski and David Damanik. The author would also like to thank William Yessen and Christoph Marx for their many helpful conversations.

%%%%%%%%%%BACKGROUND:SUBSTITUTION
%%%%%%%%%%BACKGROUND:SUBSTITUTION
%%%%%%%%%%BACKGROUND:SUBSTITUTION
%%%%%%%%%%BACKGROUND:SUBSTITUTION
%%%%%%%%%%BACKGROUND:SUBSTITUTION

\section{Background on Symbolic and Combinatorial Properties}\label{S:Background}
\subsection{Substitution Maps}\label{S:Substitution}
Let $\calA$ be an \defword{alphabet} of finitely many symbols, and let $\calA^*$ and $\tilde{\calA}$ denote the free monoid and free group generated by $\calA$, respectively. We will restrict ourselves to the case that $|\calA|=2$, in which case we will denote its members $\calA=\set{\bfzero, \bfone}$. Any $\Psi \in \End(\tilde{\calA})$ is determined by the images of $\bfzero$ and $\bfone$, which we will denote $w_\bfzero$ and $w_\bfone$, respectively. If $w \in \tilde{\calA}$, we say that $w$ is a \defword{word} and define $|w|_\bfa$ to be the sum of the powers of $\bfa$, i.e. the number of occurrences of the letter $\bfa$ in $w$ minus the number of occurrences of $\bfa\inv$ in $w$. In this way, one may associate $w$ with an element of $\Z^2$ and thus, $\Psi$ with a matrix
$$M_\Psi=\begin{pmatrix}|w_\bfzero|_\bfzero & |w_\bfone|_\bfzero\\ |w_\bfzero|_\bfone&|w_\bfone|_\bfone\end{pmatrix}$$
The map $\Psi \mapsto M_\Psi$ induces a homomorphism from $\Aut(\tilde{\calA}) \to \GL(2,\Z)$. Observe that if $\Psi$ is invertible then $\det(M_\Psi)=\pm 1$.

In the case that $w_\bfzero, w_\bfone \in \calA^*$, the map $s: \calA \to \calA^*$ given by $s:\bfzero \mapsto w_\bfzero$, $\bfone \mapsto w_\bfone$ is called a \defword{substitution}. This extends uniquely to an endomorphism of $\tilde{\calA}$ which, by abuse of notation, we also call a substitution and denote $s$. If $s$ is an automorphism of $\tilde{\calA}$, we will say that $s$ is an \defword{invertible} substitution. If there is $k$ such that $(M_s)^k$ is strictly positive, we will say that $s$ is a \defword{primitive} substitution. We call $M_s$ a \defword{substitution matrix} and look at the space of matrices corresponding to invertible substitutions
$$\calM=\left\{\begin{pmatrix}p&q\\r&s\end{pmatrix}:p,q,r,s \in \N \textnormal{ and } ps-qr=\pm 1\right\} \subset \GL(2,\Z).$$

\begin{lemma}[Wen, Wen 1994 \cite{WW94}]\label{L:Matrix_Generators}
$\calM$ is generated as a monoid by
$$h_1=\begin{pmatrix}0&1\\1&0\end{pmatrix}, \qquad h_2=\begin{pmatrix}1&1\\1&0\end{pmatrix}.$$\end{lemma}

\begin{theorem}[Wen, Wen 1994 \cite{WW94}]\label{T:InvertSubs_Generators}
The monoid of invertible substitutions is generated by $\pi$, $\sigma$, and $\rho$ where
$$\pi: \bfzero \mapsto \bfone, \bfone \mapsto \bfzero \quad \sigma: \bfzero \mapsto \bfzero\bfone, \bfone \mapsto \bfzero \quad \rho: \bfzero \mapsto \bfone\bfzero, \bfone \mapsto \bfzero.$$
\end{theorem}

\begin{remark}\label{R:SubsMatrix_Generators}
Note that $M_\pi=h_1$ and $M_\sigma=M_\rho=h_2$. The non-uniqueness of this matrix representation leads one to ask questions answered by the next theorem.
\end{remark}

\begin{theorem}[Peyri\`ere, Wen, Wen 1993 \cite{PWW93}]\label{T:Kernel=InnerAut}
An element of the kernel of the map $\Aut(\tilde{\calA}) \to \GL(2,\Z)$, $\Psi \mapsto M_\Psi$ is either an inner automorphisms, $$\iota_w: \bfzero \mapsto w\bfzero w\inv, \bfone \mapsto w\bfone w\inv \text{ for } w \in \tilde{\calA}$$  or an inner automorphism composed with the involution $\bfzero \mapsto \bfzero \inv, \bfone \mapsto \bfone \inv$.
\end{theorem}

\begin{remark}\label{R:LocIso}
As a consequence of the above, if $M_s=M_{s'}$, then $\exists w \in A^*$ such that $s'=\iota_{w}s$ or $\iota_{w \inv} s$. If this is the case, we will say $s$ and $s'$ are \defword{conjugate}. We say two sequences $u, v \in \calA^\N$ are \defword{locally isomorphic} if whenever $w$ is a word of $u$, $w$ is a word of $v$. If $s$ and $s'$ are conjugate and have fixed points $u$ and $v \in \calA^\N$, respectively, then $u$ and $v$ are locally isomorphic (see~\cite{Fog02}). We say that a fixed point $u \in \calA^\N$ of a substitution is a \defword{substitution sequence}.
\end{remark}

\begin{remark}\label{R:SubstitutionFactorization}
It can be extracted from the proof of Lemma~\reff{L:Matrix_Generators} that a substitution matrix can be decomposed as blocks of $h_2h_1$, possibly separated by $h_1$, and $h_2h_1$ is the substitution matrix for $\sigma \circ \pi$ or $\rho \circ \pi$. Hence $(h_2h_1)^n$ corresponds to
$$(\sigma \circ \pi)^{l_1} \circ (\rho \circ \pi)^{m_1} \circ ... \circ (\sigma \circ \pi)^{l_k} \circ (\rho \circ \pi)^{m_k}$$
where $\sum (l_i + m_i)=n$. We denote a sequence of this form as $\delta^{(n)}$.
\end{remark}

\begin{lemma}[Tan, Wen 2003 \cite{TW03}]\label{L:PrimInvertSubs_Decomp}
If $s$ is a primitive invertible substitution and is of the form $\delta^{(n_1)} \circ \pi \circ ... \circ \pi \circ \delta^{(n_k)}$ where $n_1, n_k \geq 0$ and $n_2,...,n_{k-1} \geq 1$, and $(1,\beta)$ is an eigenvector of $M_s$ where $\beta>0$, then 
$$\beta=\cfrac{1}{n_1+\cfrac{1}{\ddots+\cfrac{1}{n_k+\frac{1}{\beta}}}}.$$
\end{lemma}

\begin{remark}
The continued fraction expansion of $\beta$ is
$$\begin{array}{rl}
[0;n_1,\overline{n_2,n_3,...n_{k-1},n_k+n_1} ] & \textnormal{ if } n_1 \geq 1,\\
\textcolor{white}{\{}[n_2; \overline{n_3,n_4,...,n_k,n_2}] & \textnormal{ if } n_1=0 \textnormal{ and } n_k \geq 1,\\
\textcolor{white}{\{}[n_2; \overline{n_3, n_4,...,n_{k-2},n_{k-1}+n_2}] & \textnormal{ if } n_1=n_k=0.
\end{array}$$
\end{remark}

%%%%%%%%%%BACKGROUND:STURMIAN
%%%%%%%%%%BACKGROUND:STURMIAN
%%%%%%%%%%BACKGROUND:STURMIAN
%%%%%%%%%%BACKGROUND:STURMIAN
%%%%%%%%%%BACKGROUND:STURMIAN

\subsection{Sturmian Sequences and Substitution Sequences}\label{S:Sturmian}
Let $u$ be a sequence in $\calA^\N$ and recall that $w \in \tilde{\calA}$ is called a \defword{word}. Define $\calL_n(u)$ be the set of all words of $u$ of length $n$ and $\calL(u):=\bigcup_{n \in \N} \calL_n(u)$. This is the \defword{language} of $u$. We say the \defword{complexity} (of $u$) is the function $p_u(n):=|\calL_n(u)|$. It is clear that $p_u(n)$ is non-decreasing. If $u$ is eventually periodic, then $\exists n$ such that $p_u(n)=p_u(n+1)$. The converse is also true. Further, if $u$ is not eventually periodic, then $p_u(1) \geq 2$. Hence, if $u$ is not eventually periodic, then $p_u(n) \geq n+1$. We say $u$ is a \defword{Sturmian sequence} if $p_u(n)=n+1\ \forall n \in \N$. In the frame of this definition, we will only work with a two letter alphabet $\calA=\set{\bfzero,\bfone}$.

If $X$ is a collection of sequences, then we say $\calL(X):=\bigcup_{x \in X} \calL(x)$. Let $S$ be the standard \defword{left shift}, i.e. $(Su)_k=u_{k+1}$ for $k \geq 0$. Equip $\calA$ with the discrete topology and $\calA^\Z$ with the product topology. Extend $u$ to left arbitrarily so that it is a bi-infinite sequence and denote this new sequence $u'$. Define the \defword{dynamical hull} of $u$ to be the set
\begin{equation}\label{E:DynamicalHull}
\Omega_u \subset \calA^\Z:=\set{x: \exists n_k \to \infty\text{ such that }\lim_{k \to \infty} S^{(n_k)}(u')=x}.
\end{equation}
For a general dynamical system, the construction in \eqref{E:DynamicalHull} is called an \defword{$\omega$-limit set}, see for example \cite{BS02}. It follows from Remark~\reff{R:LocIso} that there is a one-to-one correspondence between substitution matrices and dynamical hulls of substitution sequences.

For the reader's convenience, we state some standard results about Sturmian sequences. Proofs of Propositions \reff{P:Frequency}, \reff{P:Freq_DynamicalHull}, \reff{P:SturmianRotationCutting}, and \reff{P:RotationCutting_DynamicalHull} can be found in \cite{Fog02}.

\begin{proposition}\label{P:Frequency}
The \defword{frequency} of \bfone's in a Sturmian sequence $u$,
$$\lim_{n \to \infty} \frac{|u_0 u_1 ... u_{n-1}|_\bfone}{n},$$
is well-defined and irrational. Further, the frequency of $u$ depends only on $\calL(u)$.
\end{proposition}

\begin{proposition}\label{P:Freq_DynamicalHull}
Let $u$ and $v$ be Sturmian sequences. Then $u$ and $v$ have the same frequency if and only if they have the same language. If this is the case, they generate the same dynamical hull.
\end{proposition}

Define the \defword{rotation sequence} of irrational angle $\alpha$ in the 1-torus $\T^1$ and initial point $\omega \in \T^1$, 
$$R_{\alpha, \omega}(n):=\chi_{[1-\alpha,1)}(n\alpha + \omega \mod 1).$$
As a matter of convention, $R_\alpha:=R_{\alpha,0}$.

Let us denote by $C_{\beta, b}$ the \defword{cutting sequence} with irrational slope $\beta$ and intercept $b$ defined as follows: In the first quadrant of $\R^2$, construct a grid with horizontal and vertical lines through each lattice point. Starting from the origin, label each intersection of $y=\beta x + b$ with a $\bfzero$ if the gridline crossed is vertical and $\bfone$ if the gridline crossed is horizontal. Again, we will use the convention that $C_\beta := C_{\beta, 0}$.

\begin{proposition}\label{P:SturmianRotationCutting}
The following are equivalent:
\begin{enumerate}
\item The sequence $u$ is a Sturmian sequence of frequency $\alpha$.
\item The sequence $u$ is in the dynamical hull of $R_\alpha$. (see also \cite{MH40})
\item The sequence $u$ is in the dynamical hull of $C_\beta$ for $\alpha=\frac{\beta}{1+\beta}$.
\end{enumerate}
\end{proposition}

\begin{proposition}\label{P:RotationCutting_DynamicalHull}\quad
\begin{enumerate}
\item The dynamical hull of $R_\alpha$ is $\bigcup_{\omega \in \T^1} R_{\alpha, \omega}$.
\item The dynamical hull of $C_\beta$ is $\bigcup_{b \in [0, 1)} C_{\beta, b}$.
\end{enumerate}
\end{proposition}

\begin{remark}\label{R:RotationCutting_Ambiguous}
The following caveat must be added to Proposition~\reff{P:RotationCutting_DynamicalHull}: A rotation sequence, $R_{\alpha, \omega}(n)$ may also be defined using $\chi_{(1-\alpha,1]}$. When we say ``$\bigcup_{\omega \in \T^1} R_{\alpha, \omega}$,'' we include this countable collection of sequences. Similarly, in the case of a cutting sequence, there is an ambiguity in the above definition if $y=\beta x + b$ crosses a lattice point. In this case, we consider both the sequence with $\bfzero\bfone$ and the sequence with $\bfone\bfzero$ at this crossing. Details of this can be found in \cite{Dam07}.
\end{remark}

The fixed point of a primitive invertible substitution is Sturmain, i.e. $C_{\beta,b}$ for some irrational $\beta$ and some $b$. Recall that substitution $s$ is called primitive if there is $k$ such that $(M_s)^k$ is strictly positive. From this, we may apply the Perron-Frobenius theorem to conclude that $M_s$ has two distinct eigenvalues $\mu>1$ and $\mu\inv$, and there is an eigenvector corresponding to $\mu$ is in the first quadrant. The slope $\beta$ is the frequency of the letter $\bfone$ per one letter $\bfzero$. Hence, it is exactly the second coordinate of the eigenvector we previously found in Lemma~\reff{L:PrimInvertSubs_Decomp}.

\begin{samepage}
\begin{proposition}[Tan, Wen 2003 \cite{TW03}]\label{P:ZeroIntercept}
Suppose the fixed point of $s$ is $C_{\beta,b}$. The following are equivalent:
\begin{enumerate}
\item $s \in \brac{\pi, \sigma}$, that is, $s$ is a product of $\pi$ and $\sigma$.
\item $b=0$.
\end{enumerate}
\end{proposition}
\end{samepage}

\begin{remark}
Due to Remark~\reff{R:LocIso} and Proposition~\reff{P:SturmianRotationCutting}, we need only classify one representative element of each dynamical hull. Proposition~\reff{P:ZeroIntercept} shows that classifying the elements $\set{C_\beta}_{\beta \in \R^{\geq 0}}$ that are fixed by a substitution effectively classifies all Sturmian sequences that are fixed by a substitution.
\end{remark}

\begin{theorem}[Crisp, Moran, Pollington, Shiue 1993 \cite{CMPS93}]\label{T:SturmianInvariant}\quad
\begin{itemize}
\item If $\beta>1$ is irrational, $C_\beta$ is invariant under some non-trivial substitution $s$ if and only if $\beta=[b_0;\overline{b_1,...,b_n}]$, where $b_n \geq b_0 \geq 1$. Further, if $n$ is minimal, then $s$ is a power of
$$\pi \circ (\sigma \circ \pi)^{b_0} \circ \pi \circ ... \circ \pi \circ (\sigma \circ \pi)^{b_{n-1}} \circ \pi \circ (\sigma \circ \pi)^{b_n-b_0} \circ \pi.$$
\item If $0<\beta<1$ is irrational, $C_\beta$ is invariant under some non-trivial substitution $s$ if and only if $\beta=[0;b_0,\overline{b_1,...,b_n}]$ where $b_n \geq b_0$. Further, if $n$ is minimal, then $s$ is a power of
$$(\sigma \circ \pi)^{b_0} \circ \pi \circ ... \circ \pi \circ (\sigma \circ \pi)^{b_{n-1}} \circ \pi \circ (\sigma \circ \pi)^{b_n-b_0}.$$\end{itemize}
\end{theorem}

\begin{remark}
Observe that these correlate exactly to the conclusion of Lemma~\reff{L:PrimInvertSubs_Decomp}. As stated in Remark \reff{R:SubstitutionFactorization}, substitutions can be factored in a canonical way and then associated to a rotation sequence. For example, the substitution 
$$\pi \circ (\sigma \circ \pi)^{n_2} \circ \pi ... \circ \pi \circ (\sigma \circ \pi)^{n_k}$$
fixes the Sturmain sequences
$$C_\beta, \quad \beta=[n_2; \overline{n_3,n_4,...,n_k,n_2}]$$
which is equal to
$$R_\alpha, \quad \alpha=[0; 1, \overline{n_2,n_3,n_4,...,n_k}].$$
\end{remark}

\begin{remark}
It is also worth noting that not all rotation sequences, even those with eventually periodic continued fraction expansions, are fixed by a substitution.
\end{remark}

\begin{theorem}[Brown 1991 \cite{Bro91}]\label{T:SturmianInvariant_CE}
Let $\alpha=[0; 5, \overline{1}]$. No non-trivial substitution fixes $R_\alpha$.
\end{theorem}

%%%%%%%%%%TRACE MAPS:SUBSTITUTION
%%%%%%%%%%TRACE MAPS:SUBSTITUTION
%%%%%%%%%%TRACE MAPS:SUBSTITUTION
%%%%%%%%%%TRACE MAPS:SUBSTITUTION
%%%%%%%%%%TRACE MAPS:SUBSTITUTION

\section{Trace Maps}\label{S:TraceMaps}
\subsection{Trace Maps Corresponding to Substitution Hamiltonians}\label{S:SubstitutionTrace}
Suppose a substitution $s$ has a fixed point $u$ and further suppose $v$ is in the dynamical hull $\Omega_u$ of this substitution sequence as described in Section~\reff{S:Sturmian}. We define a substitution Hamiltonian $H$ with \defword{coupling constant} $\lambda>0$ to be the following bounded self-adjoint operator on $\ell^2(\Z)$,
\begin{equation}\label{E:SubstitutionOperator}
(H_\lambda \psi)(n)=\psi(n-1)+\psi(n+1)+\lambda \cdot v(n) \cdot \psi(n).
\end{equation}
Note that the operator in \eqref{E:SubstitutionOperator} depends on the substitution $s$, but we omit $s$ from the notation as the substitution map remains fixed. Having fixed $s$, $H$ is independent of the choice of $v$ due to the minimality of $(\Omega_u, S)$ from \eqref{E:DynamicalHull}, see for example \cite{Dam00}.

A standard tool for studying the spectra of such operators is the \defword{transfer matrix} $A$, which we derive from the formal eigenvalue equation
\begin{equation}\label{E:EigenSubstitution}
\psi(n+1)+\psi(n-1)+\lambda \cdot v(n)\cdot\psi(n)=E\cdot\psi(n).
\end{equation}
Observe that $\psi$ satisfies \eqref{E:EigenSubstitution} if and only if the following holds:
\begin{equation*}
\overbrace{\begin{pmatrix}
E-\lambda \cdot v(n) & -1\\
1& 0
\end{pmatrix}}^{A(\lambda \cdot v(n),E)}
\begin{pmatrix}
\psi(n)\\
\psi(n-1)
\end{pmatrix}
=
\begin{pmatrix}
\psi(n+1)\\
\psi(n)
\end{pmatrix}
\end{equation*}
Furthermore, we must have:
\begin{equation*}
%\prod_{k=n-1}^{0} A(\lambda \cdot v(k),E)
A(\lambda \cdot v(n-1),E) \circ ... \circ A(\lambda \cdot v(0),E)
\left(\begin{array}{c}
\psi(0)\\
\psi(-1)
\end{array}\right)
=
\left(\begin{array}{c}
\psi(n)\\
\psi(n-1)
\end{array}\right).
\end{equation*}

\begin{definition}\label{D:CocycleSubs}
Let $\calA=\set{\bfzero,\bfone}$. Formally, define a map $A(\lambda,E):\mathcal{A} \to \SL(2, \mathbb{R})$
$$A(\lambda,E)(\bfa) := A(\lambda \cdot \bfa, E) = 
\left(\begin{array}{cc}
E-\lambda \cdot \bfa & -1\\
1 & 0
\end{array} \right) \text{ for } \bfa \in \calA
$$
and extend this to a mapping on $\mathcal{A}^*$ by
$$A(\lambda,E)(\bfa_1...\bfa_n)=A(\lambda,E)(\bfa_n) \cdot ... \cdot A(\lambda,E)(\bfa_1).$$
The induced map $$(S, A(\lambda,E)): \Omega_u \times \R^2 \to \Omega_u \times \R^2 \text{ given by}$$
$$\left(v, \left(\begin{array}{c}
\psi(0)\\
\psi(-1)
\end{array}\right)\right) \mapsto \left(Sv, A(\lambda,E)(v(0))
\left(\begin{array}{c}
\psi(0)\\
\psi(-1)
\end{array}\right)
\right)$$
form a family of $\SL(2, \R)$ cocycles associated with $H_\lambda$. These are referred to in the literature as \defword{Schr\"odingier cocycles}. 
\end{definition}

\begin{definition}\label{D:TraceMapSubs}
Define
\begin{equation*}x_k(E)=\frac{1}{2} \tr (A(\lambda,E)(s^k(\bfzero\bfone))),\end{equation*}
\begin{equation*}y_k(E)=\frac{1}{2} \tr (A(\lambda,E)(s^k(\bfone))),\end{equation*}
\begin{equation*}z_k(E)=\frac{1}{2} \tr (A(\lambda,E)(s^k(\bfzero))).\end{equation*}
with the convention that $s^0$ is the identity substitution. We leave the dependence on $\lambda$ implicit for now. For any invertible substitution, one may derive an invertible polynomial mapping $T_s: \mathbb{R}^3 \to \mathbb{R}^3$ called the \defword{trace map} such that
$$T_s: \R^3 \to \R^3, \quad
\begin{pmatrix} x_k(E)\\y_k(E)\\z_k(E) \end{pmatrix} \mapsto
\begin{pmatrix} x_{k+1}(E)\\y_{k+1}(E)\\z_{k+1}(E) \end{pmatrix}.$$
\end{definition}
For more background on a trace map associated with a substitution, see \cite{Dam00}.

%%%%%%%%%%TRACE MAPS:STURMIAN
%%%%%%%%%%TRACE MAPS:STURMIAN
%%%%%%%%%%TRACE MAPS:STURMIAN
%%%%%%%%%%TRACE MAPS:STURMIAN
%%%%%%%%%%TRACE MAPS:STURMIAN

\subsection{Trace Maps Corresponding to Sturmian Hamiltonians}\label{S:SturmianTrace}

Proceeding as before, let $\lambda>0$, $\alpha \in \T^1$ irrational, and $\omega \in \T^1$. We consider the bounded self-adjoint operator on $\ell^2(\Z)$ defined at the beginning of this paper,
\begin{equation*}\tag{\reff{E:SchrodingerOperator}}
(H_{\lambda, \alpha, \omega} \psi)(n)=\psi(n+1)+\psi(n-1)+\lambda \cdot R_{\alpha, \omega}(n) \cdot\psi(n),\quad \psi \in \ell^2(\Z),
\end{equation*}
We say this is the \defword{Sturmian Hamiltonian} with \defword{potential} given by $R_{\alpha, \omega}$.
See \cite{BIST89} for background on the trace map associated with an irrational rotation. As was previously mentioned in Section~{\reff{S:Introduction}}, we will denote its spectrum by $\Sigma_{\lambda, \alpha}$.
Let 
\begin{equation}\label{E:ContFrac}
\alpha=[0; a_1, a_2, ... ]
\end{equation}
be the continued fraction expansion of $\alpha$ and let $p_n/q_n=[0; a_1, ... , a_n]$ be the $n\th$ rational approximant to $\alpha$. It is a standard result in continued fractions (see, for example, \cite{Khi97}) that
\begin{equation}\label{E:RationalApproximantP}
p_{k+1}=a_{k+1}p_k + p_{k-1},\quad p_1 = 1,\quad  q_0 = 0
\end{equation}
\begin{equation}\label{E:RationalApproximantQ}
q_{k+1}=a_{k+1}q_k + q_{k-1},\quad q_1 = a_1,\quad  q_0 = 1
\end{equation}

\begin{definition}\label{D:TraceMap}
For $\lambda>0$, $\alpha \in \T^1$ irrational, and $k \geq 0$, define
\begin{equation*}x_k (E) = \frac{1}{2} \tr (A_k(E)A_{k-1}(E)),\end{equation*}
\begin{equation*}y_k (E) = \frac{1}{2} \tr (A_k(E)),\end{equation*}
\begin{equation*}z_k (E) = \frac{1}{2} \tr (A_{k-1}(E)).\end{equation*}
where in keeping of the notation of Definition \reff{D:CocycleSubs}
$$\begin{array}{RCL}
%A(\lambda, E) &=&	\begin{pmatrix}
%			E-\lambda & -1\\
%			1 & 0
%			\end{pmatrix}\\
A_k(E) &=&	\left\{ \begin{array}{ll}
			A(\lambda, E)(R_\alpha(q_k)) \cdot ... \cdot A(\lambda, E)(R_\alpha (1)) & k \geq 1\\
			\begin{pmatrix}
			E & -1\\
			1 & 0
			\end{pmatrix}, & k = 0\\
			\begin{pmatrix}
			1 & -\lambda\\
			1 & 0
			\end{pmatrix}, & k=-1.
			\end{array} \right.
\end{array}$$
In direct analogy to Definition~\reff{D:TraceMapSubs}, let us to define the \defword{trace map}
$$T_\alpha: \R^3 \to \R^3, \quad
\begin{pmatrix} x_k(E)\\y_k(E)\\z_k(E) \end{pmatrix} \mapsto
\begin{pmatrix} x_{k+1}(E)\\y_{k+1}(E)\\z_{k+1}(E) \end{pmatrix}.$$
\end{definition}

\begin{remark}\label{R:TraceMap}
Definition~\reff{D:TraceMap} is a direct generalization of Definition \reff{D:CocycleSubs}. Note that Equation (\reff{E:RationalApproximantQ}) is reflected in the recursion $A_{k+1}(E)=A_{k-1}(E)A_k(E)^{a_{k+1}}$. We leave the dependence of $x_k(E)$, $y_k(E)$, $z_k(E)$, and $A_k(E)$ on $\lambda$ and $\alpha$ implicit. We remark also that $T_\alpha$ is formally not a mapping as each iteration is dependent on $k$, but omit this from the notation. 
\end{remark}

We state and prove the following well-known lemma for the reader's convenience.

\begin{lemma}\label{L:Chebyshev}
For $A \in \SL(2,\R)$ and $a \in \N$,
\begin{equation}\label{E:ChebyshevBaseCase} A^2=\tr(A)A-I,\end{equation}
\begin{equation*} A^a = U_{a-1}(x)\cdot A - U_{a-2}(x) \cdot I, \end{equation*}
where $x=\frac{1}{2}\tr(A)$ and $U_a(x)$ are Chebyshev polynomials of the second kind given by the following recursive relationship:
$$\begin{array}{RCL}
U_{-1}(x)&=& 0\\
U_{0}(x)&=& 1\\
U_{a+1}(x)&=& U_a(x) \cdot 2x - U_{a-1}(x).
\end{array}$$
\end{lemma}

\begin{proof}
Since $A$ is unimodular, Equation (\reff{E:ChebyshevBaseCase}) follows by a direct application of the Cayley-Hamilton theorem. Proceeding by induction, we have that
$$\begin{array}{RCL}
A^{a+1} &=& A \cdot A^a\\
&=& A\left(U_{a-1}\left(\frac{1}{2}\tr(A)\right)\cdot A - U_{a-2}\left(\frac{1}{2}\tr(A)\right) \cdot I\right)\\
&=& U_{a-1}\left(\frac{1}{2}\tr(A)\right)\cdot A^2 - U_{a-2}\left(\frac{1}{2}\tr(A)\right) \cdot A\\
&=& U_{a-1}\left(\frac{1}{2}\tr(A)\right)\left(2\cdot\frac{1}{2}\tr(A)A-I\right) - U_{a-2}\left(\frac{1}{2}\tr(A)\right) \cdot A\\
&=& \left(U_{a-1}\left(\frac{1}{2}\tr(A)\right)\left(2\cdot\frac{1}{2}\tr(A)\right) - U_{a-2}\left(\frac{1}{2}\tr(A)\right)\right) \cdot A - U_{a-1}\left(\frac{1}{2}\tr(A)\right) \cdot I\\
&=& U_{n}(x)\cdot A - U_{a-1}(x) \cdot I
\end{array}$$
\end{proof}

\begin{definition}\label{D:T_a}
We define, for $a \in \N$,
$$ T_a: \mathbb{R}^3 \to \mathbb{R}^3, \quad 
\begin{pmatrix} x\\ y\\ z \end{pmatrix} \to
\begin{pmatrix}x U_a (y) - z U_{a-1}(y)\\xU_{a-1}(y) - z U_{a-2}(y) \\ y \end{pmatrix}.$$
\end{definition}

\begin{definition}\label{D:UP}
We also define auxiliary maps
$$U: \mathbb{R}^3 \to \mathbb{R}^3, \quad 
\begin{pmatrix} x\\ y\\ z \end{pmatrix} \to
\begin{pmatrix} 2xz- y \\ x \\ z \end{pmatrix},$$
$$P: \mathbb{R}^3 \to \mathbb{R}^3, \quad 
\begin{pmatrix} x\\ y\\ z \end{pmatrix} \to
\begin{pmatrix}x \\ z \\ y \end{pmatrix}.$$
\end{definition}

\begin{lemma}\label{L:TraceMap_Generators}
$UP=T_1$ and $T_a=UT_{a-1}=U^{a-1}T_1=U^aP$.
\end{lemma}

\begin{proof}
This can be verified through straightforward calculation. Alternatively, recalling the generators for invertible substitutions, one could note that $U=F_{\sigma \circ \pi}$. U is also $F_{\rho \circ \pi}$, but we prefer to work with $\sigma$ due to Proposition~\reff{P:ZeroIntercept} and $P=F_\pi$ and $s_a=\pi \circ (\sigma \circ \pi)^a$.
\end{proof}

\begin{remark}\label{R:T_a}
Observe that $T_1$ is the trace map for rotation by the golden mean $\frac{\sqrt{5}-1}{2}$. As a substitution sequence, this is obtained by $s_1: \bfzero \mapsto \bfone, \bfone \mapsto \bfone\bfzero$. In general, $T_a$ is the trace map for rotation by the $a\th$ metallic mean $\alpha=[0,\overline{a}]$, or equivalently, the substitution sequence generated by $s_a: \bfzero \mapsto \bfone, \bfone \mapsto \bfone^a\bfzero$. It is clear that these substitutions are primitive and invertible. Thus this is a generalization of the trace map presented in the previous subsection.
\end{remark}

\begin{lemma}\label{L:DT_a}
We have
$$DT_a \begin{pmatrix} x\\y\\z \end{pmatrix} =
\begin{pmatrix}
U_a(y) & xU_a(y)-zU'_{a-1}(y) & -U_{a-1}(y)\\
U_{a-1}(y) & xU_{a-1}(y) - z U'_{a-2}(y) & -U_{a-2}(y)\\
0 & 1 & 0
\end{pmatrix}.$$
In particular
$$DT_a \begin{pmatrix} 1\\1\\1 \end{pmatrix}=
\begin{pmatrix}
a+1 & \frac{(a+1)^3 - (a+1)}{3} - \frac{a^3-a}{3} & -a \\
a & \frac{a^3-a}{3} - \frac{(a-1)^3 - (a-1)}{3} & -a+1 \\
0 & 1 & 0
\end{pmatrix}=
\begin{pmatrix}
a+1 & a^2+a & -a \\
a & a^2-a & -a+1 \\
0 & 1 & 0
\end{pmatrix}
$$
which has eigenvalues
$$\frac{a^2}{2} - \frac{a \sqrt{a^2+4}}{2} + 1, \quad \frac{a^2}{2}+\frac{a\sqrt{a^2+4}}{2} + 1, \quad -1$$
and corresponding eigenvectors
$$\begin{pmatrix}
\frac{a-2}{a} + \frac{(a+2)(\frac{a^2}{2}-\frac{a\sqrt{a^2+4}}{2}+1}{a}\\ \frac{a^2}{2}-\frac{a\sqrt{a^2+4}}{2}+1\\1
\end{pmatrix},\quad
\begin{pmatrix}
\frac{a-2}{a} + \frac{(a+2)(\frac{a^2}{2}+\frac{a\sqrt{a^2+4}}{2}+1}{a}\\ \frac{a^2}{2} + \frac{a\sqrt{a^2+4}}{2}+1\\1
\end{pmatrix}, \quad
\begin{pmatrix}
a \\ -1 \\ 1
\end{pmatrix}.$$
\end{lemma}

\begin{proof}
The form of $DT_a(x,y,z)$ follows directly from the definition of $T_a$. To compute $DT_a(1,1,1)$, it is useful to note that $U_a(1)=a+1$ and $U'_a(1)=\frac{(a+1)^3-(a+1)}{3}$, which follows directly from the definition of $U_a$. Lastly, the eigenvectors and eigenvalues can be found using a straightforward computation.
\end{proof}

We also state for the reader's convenience that
$$DU(P_1)=\begin{pmatrix} 2 & -1& 2\\1 & 0 & 0\\ 0 & 0 & 1 \end{pmatrix},\qquad
D(UP)(P_1)=\begin{pmatrix}2&2&-1\\1&0&0\\0&1&0\end{pmatrix}.$$

%We claim that $T_\alpha$ is the composition of $T_a$ as follows:
%$$\begin{pmatrix}x_k(E)\\y_k(E)\\z_k(E)\end{pmatrix}=T_{a_k} \circ ... \circ T_{a_1} \begin{pmatrix}(E-\lambda)/2\\ E/2\\ 1\end{pmatrix},$$
%where $\alpha=[0; a_1, a_2,...]$ as in Equation~\reff{E:ContFrac}. 
%We restate this in Lemma~\reff{L:InitialConditions} and present a proof by direct calculation.
%As stated in Remark \reff{R:SubstitutionFactorization}, substitutions can be factored in a canonical way and then associated to a rotation sequence. For example, the substitution 
%$$\pi \circ (\sigma \circ \pi)^{n_2} \circ \pi ... \circ \pi \circ (\sigma \circ \pi)^{n_k}$$
%fixes the Sturmain sequences
%$$C_\beta, \quad \beta=[n_2; \overline{n_3,n_4,...,n_k,n_2}]$$
%which is equal to
%$$R_\alpha, \quad \alpha=[0; 1, \overline{n_2,n_3,n_4,...,n_k}].$$

%%%%%%%%%%TRACE MAPS:BACKGROUND
%%%%%%%%%%TRACE MAPS:BACKGROUND
%%%%%%%%%%TRACE MAPS:BACKGROUND
%%%%%%%%%%TRACE MAPS:BACKGROUND
%%%%%%%%%%TRACE MAPS:BACKGROUND

\subsection{Properties of the Trace Map and Previous Results}\label{S:PreviousResults}
\begin{definition}
Let us define the Fricke-Vogt invariant
$$I(x,y,z)=x^2+y^2+z^2-2xyz-1$$
and level surfaces
$$S_\lambda=\left\{(x,y,z) \in \R^3 : I(x,y,z) = \frac{\lambda^2}{4} \right\}.$$
\end{definition}

\begin{lemma}\label{L:InvariantSurfaceEntire}
$S_\lambda$ is invariant under a trace map $T_a$ for $a \in \N$.
\end{lemma}
\begin{proof}
This is readily checked via direction computation and Lemma~\reff{L:TraceMap_Generators} as $I(x,y,z)$ is preserved by $U$ and $P$.
\end{proof}

\begin{remark}
Here, we state some well-known and easily verified facts about $S_\lambda$: The level sets $I(x,y,z)=I_0$ for $I_0 \in \R$ foliate $\R^3$. The surface $S_0$ is known as the \defword{Cayley cubic} and consists of a central portion with four conic singularities at $P_1=(1,1,1)$, $P_2=(-1,-1,1)$, $P_3=(-1,1,-1)$, and $P_4=(1,-1,-1)$ and four unbounded regions emanating from these singularities. Denote by $\S$ the central portion of $S_0$ inside the cube $\set{|x| \leq 1, |y| \leq 1, |z| \leq 1}$. For $\lambda>0$, $S_\lambda$ is connected, smooth, and non-compact.
\end{remark}

Recall from Definition~\reff{D:TraceMap} and Remark~\reff{R:TraceMap} the trace map $T_\alpha$ of $H_{\lambda, \alpha}$:
$$T_\alpha: \R^3 \to \R^3, \quad
\begin{pmatrix} x_k(E)\\y_k(E)\\z_k(E) \end{pmatrix} \mapsto
\begin{pmatrix} x_{k+1}(E)\\y_{k+1}(E)\\z_{k+1}(E) \end{pmatrix}.$$

\begin{lemma}\label{L:InitialConditions}
Fix $\lambda$ and define a line of initial conditions
\begin{equation}\label{E:InitCondSchrodinger}
\ell_\lambda:=\left\{\left(\frac{E-\lambda}{2}, \frac{E}{2}, 1 \right): E \in \R\right\}.
\end{equation}
\begin{enumerate}
\item\label{L:CompositionDefinition} We claim that $T_\alpha$ is the composition of $T_a$ as follows:
$$\begin{pmatrix}x_k(E)\\y_k(E)\\z_k(E)\end{pmatrix}=T_{a_k} \circ ... \circ T_{a_1} \begin{pmatrix}(E-\lambda)/2\\ E/2\\ 1\end{pmatrix}$$
\item\label{L:InvariantSurface} For $E \in \R$ and $k \in \N$, $(x_k(E), y_k(E), z_k(E))$ lies in the surface $S_\lambda$.
\item\label{L:BoundedOrbit} An energy $E$ belongs to the spectrum $\Sigma_{\lambda, \alpha}$ if and only if the sequence $(x_k(E), y_k(E), z_k(E))$ remains bounded as $k \to +\infty$.
\end{enumerate}
\end{lemma} 

\begin{proof}
To prove Part~\reff{L:CompositionDefinition}, observe that $\ell_\lambda$ corresponds to $k=0$. So we will achieve our desired result by induction if we show
$$\begin{pmatrix}x_{k+1}(E)\\y_{k+1}(E)\\z_{k+1}(E)\end{pmatrix}=T_{a_{k+1}}\begin{pmatrix}x_k(E)\\y_k(E)\\z_k(E)\end{pmatrix}.$$
The first two coordinates follow from our earlier remark that $$A_{k+1}(E)=A_{k-1}(E)A_k(E)^{a_{k+1}}$$ and $z_{k+1}$ is clearly equal to $y_k$ by definition. Part~\reff{L:InvariantSurface}, follows from the observation that the line $\ell_\lambda$ lies in $S_\lambda$, Part~\reff{L:CompositionDefinition}, Lemma~\reff{L:TraceMap_Generators}, and Lemma~\reff{L:InvariantSurfaceEntire}. Finally, refer to \cite{BIST89} and \cite{Dam00} for Part~\reff{L:BoundedOrbit}.
\end{proof}

\begin{lemma}\label{L:Semi-conjugacy}
Under the semi-conjugacy
\begin{equation}\label{E:SemiConj}
F:\begin{pmatrix}\theta \\ \phi \end{pmatrix} \mapsto \begin{pmatrix} \cos(2 \pi (\theta + \phi)) \\ \cos(2\pi\theta) \\ \cos(2\pi\phi)\end{pmatrix},
\end{equation}
\begin{enumerate}
\item $U|_{\S}$ is a factor of the $\T^2 \to \T^2$ map given by $\begin{pmatrix}1&1\\0&1\end{pmatrix}$.
\item $P|_{\S}$ is a factor of the $\T^2 \to \T^2$ map given by $\begin{pmatrix}0&1\\1&0\end{pmatrix}$.
\item\label{L:Semiconj_Compposition} $T_a|_{\S}$ is a factor of the $\T^2 \to \T^2$ map given by $\begin{pmatrix}a&1\\1&0\end{pmatrix}$.
\end{enumerate}
\end{lemma}

\begin{proof}
The following uses the sum-to-angle formula and the fact that cosine is even.

$$\begin{array}{RCL}
U(F(\theta,\phi))&=&U(\cos(2 \pi (\theta + \phi)),\cos(2\pi\theta),\cos(2\pi\phi))\\
&=&(2\cos(2 \pi (\theta + \phi)\cos(2\pi\phi)-\cos(2\pi\theta),\cos(2 \pi (\theta + \phi)), \cos(2\pi\theta))\\
&=&(\cos(2\pi(\theta+2\phi)),\cos(2\pi(\theta+\phi)),\cos(2\pi\phi))\\
&=&F(\theta+\phi, \phi)\\
&=&F\left(\begin{pmatrix}1&1\\0&1\end{pmatrix}\begin{pmatrix}\theta \\ \phi \end{pmatrix}\right)\\
\\
P(F(\theta,\phi))&=&P(\cos(2 \pi (\theta + \phi)),\cos(2\pi\theta),\cos(2\pi\phi))\\
&=&(\cos(2 \pi (\phi + \theta)), \cos(2 \pi \phi), \cos(2 \pi \theta))\\
&=&F(\theta,\phi)\\
&=&F\left(\begin{pmatrix}0&1\\1&0\end{pmatrix}\begin{pmatrix}\theta \\ \phi \end{pmatrix}\right)
\end{array}$$
Part~\reff{L:Semiconj_Compposition} follows from the above and Lemma~\reff{L:TraceMap_Generators}.
\end{proof}

\begin{definition}
Denote this matrix $\begin{pmatrix}a&1\\1&0\end{pmatrix}$ by $M_a$.
\end{definition}

\begin{remark}
It is worthwhile to observe at this point that, motivated by the proof of Lemma~\reff{L:TraceMap_Generators} and Proposition~\reff{P:ZeroIntercept}, that one may directly compute the following three substitution matrices.
$$\sigma \circ \pi: \bfzero \mapsto \bfzero, \bfone \mapsto \bfzero1, \quad
M_{\sigma \circ \pi}=\begin{pmatrix} 1 & 1\\ 0 & 1 \end{pmatrix}$$
$$\pi: \bfzero \mapsto \bfone, \bfone \mapsto \bfzero \quad
M_{\pi}=\begin{pmatrix} 0 &1\\1&0\end{pmatrix}$$
$$(\sigma \circ \pi)^a: \bfzero \mapsto \bfzero, \bfone \mapsto \bfzero^a\bfone, \quad
M_{(\sigma \circ \pi)^n}=\begin{pmatrix} 1 & n\\ 0 & 1 \end{pmatrix}$$
\end{remark}

%%%%%%%%%%HyperbolicityNW
%%%%%%%%%%HyperbolicityNW
%%%%%%%%%%HyperbolicityNW
%%%%%%%%%%HyperbolicityNW
%%%%%%%%%%HyperbolicityNW
\section{Hyperbolicity of the Nonwandering Set}\label{S:HyperbolicityNW}

Recall that we are interested in irrational numbers with eventually periodic continued fraction expansions. Let us introduce the convention 
\begin{equation}
\alpha=[0;a_1...a_m,\overline{a_{m+1}...a_{m+n}}],
\end{equation} that is the continued fraction expansion of $\alpha$ has a non-repeating segment of length $m$ and a repeating segment of length $n$. Let us introduce the notation $$T_\alpha^{\text{int}}:=T_{a_m} \circ ... \circ T_{a_1}$$ and $$T_\alpha^{\text{per}}=T_{a_{m+n}} \circ ... \circ T_{a_{m+1}}$$ with the understanding that if $\alpha$ has a periodic continued fraction expansion, then $m=0$ and $T_\alpha^{\text{int}}:=\text{id}$. We will also say that $M_\alpha^\text{int}=M_{a_m} \circ ... \circ M_{a_1}$ and $M_\alpha^\text{per}=M_{a_{m+n}} \circ ... \circ M_{a_{m+1}}$ so that via Lemma~\reff{L:Semi-conjugacy}, $F(M_\alpha^\text{int})=T_\alpha^\text{int}$ and $F(M_\alpha^\text{per})=T_\alpha^\text{per}$.

We are interested in studying the set of points in each level surface $S_\lambda$ with bounded positive semiorbits under the trace map $T_\alpha$, call this set $\B^+_{\lambda, \alpha}$. Using the notation of Part~\reff{L:CompositionDefinition} of Lemma~\reff{L:InitialConditions}, $\B^+_{\lambda, \alpha}$ is the set
$$\set{\begin{pmatrix}x_0(E)\\y_0(E)\\z_0(E)\end{pmatrix}\ :\ \set{T_{a_k} \circ ... \circ T_{a_1} \begin{pmatrix}x_0(E)\\y_0(E)\\z_0(E)\end{pmatrix}\ : \ k \in \N} \text{ is bounded }}$$
Due to Lemma \reff{L:InitialConditions} Part \reff{L:BoundedOrbit}, the spectrum is exactly the set of energies corresponding to points in $\ell_\lambda \cap \B^+_{\lambda, \alpha}$. We will say that a point is \defword{nonwandering} with respect to a diffeomorphism $f:M \to M$ if for any open set $U$ containing $x \in M$, there is a $k>0$ such that $f^k(U) \cap U \neq \emptyset$. Denote the set of nonwandering points of $T_\alpha^\text{per}$ in a fixed level surface $S_\lambda$ by $\Omega_{\lambda, \alpha}^\text{per}$. Denote by $\Omega_{\lambda, \alpha}$ the preimage under $T_\alpha^\text{int}$ of $\Omega_{\lambda, \alpha}^\text{per}$.

An invariant closed set $\Lambda$ of a diffeomorphism $f:M \to M$ is \defword{hyperbolic} if for every $x \in \Lambda$, there is an invariant splitting of the tangent space into stable and unstable subspaces. That is: $T_x(M)=E^u(x) \oplus E^s(x)$, the differential $Df$ contracts vectors from the stable subspace $E^s(x)$ and $Df \inv$ contracts vectors from the unstable subspace $E^u(x)$. %The stable set (resp. unstable set) of $x$, denoted $W^s(x)$ (resp. $W^u(x)$), is the collection of points which converge to $x$ under iterations of $f$ (resp. $f \inv$). The stable set of a hyperbolic set $\Lambda$, denoted $W^s(\Lambda)$ and defined to be $\bigcup_{x \in \Lambda} W^s(x)$.
For $x \in \Lambda$ and $\epsilon>0$, define the local stable set
$$W^s_\epsilon=\{y \in M : d(f^n(x), f^n(y)) \leq \epsilon\ \forall n \in \N\}.$$
The local unstable set $W^u_\epsilon$ is defined similarly for $f\inv$.
Define the global stable set of a point $x$ to be
$$W^s(x)=\bigcup_{n \in \N} f^{-n}(W^s_\epsilon(x))$$
and define the stable set of a hyperbolic set $\Lambda$ to be
$$W^s(\Lambda)=\bigcup_{x \in \Lambda} W^s(x).$$
The definition of the global unstable set of a point $x$ and unstable set of a hyperbolic set $\Lambda$ are defined similarly.

\begin{theorem}[Theorem 1.2 from \cite{Can09}]\label{T:Cantat}
If $\alpha$ has eventually periodic continued fraction expansion, then $T_\alpha^{\text{per}}$ is pseudo-Anosov on $\S$ and for each fixed $\lambda>0$, the dynamics of $T_\alpha^{\text{per}}$ on the set of points whose full orbits are bounded under $T_\alpha^{\text{per}}$, is a locally maximal, compact, transitive hyperbolic set of $T_\alpha^\text{per}$.
\end{theorem}

\begin{remark}\label{R:StableLamination}
The set of all points with bounded positive semiorbits under $T_\alpha^\text{per}$ is exactly $T_\alpha^\text{int}(\B^+_{\lambda,\alpha})$ in a fixed level surface $S_\lambda$ defined in the preceding paragraph. (See \cite{Can09}). %We define $W^S(\B_{\lambda, \alpha})$ to be $T_\alpha^\text{int}(\B_{\lambda,\alpha})$.
\end{remark}

\begin{lemma}\label{L:Bounded=NW}
For a fixed $\lambda$ and $\alpha$, the set of points whose full orbits are bounded under $T_\alpha^{\text{per}}$ in a fixed level surface $S_\lambda$ is equal to  $\Omega_{\lambda, \alpha}^\textrm{per}$.
\end{lemma}

\begin{proof}
It can be extracted from \cite{Rob96} that escape is an open condition, more precisely any sequence of points $\set{(x_n, y_n, z_n)}_{n \in \N}$ diverging to $\infty$ on $S_\lambda$ must enter an open region $E$. Any unbounded orbit of has a subsequence which diverges to infinity, thus this orbit must enter $E$ and escape to $\infty$. Further, since $E$ is open, a neighborhood of $\mathbf{v}_0$ also escapes implying that $\mathbf{v}_0$ is not nonwandering.

In the other direction, the hyperbolicity of the set of points whose full orbits are bounded under $T_\alpha^{\text{per}}$ from \cite{Can09} implies that every point of this set is nonwandering.
\end{proof}

\begin{corollary}\label{C:Cantat}[to Theorem \reff{T:Cantat}] For fixed $\alpha$ with eventually periodic continued fraction expansion and $\lambda>0$,\begin{enumerate}
\item The set $\Omega_{\lambda, \alpha}$ is homeomorphic to a Cantor set.
\item (Theorem 6.5 from \cite{Can09}) Hausdorff dimension $\dim_\text{H} \Sigma_{\lambda, \alpha}$ is an analytic function of $\lambda$ that is strictly smaller than one and strictly greater than zero.
\end{enumerate}
\end{corollary}

%%%%%%%%%%PERIODIC CURVE
%%%%%%%%%%PERIODIC CURVE
%%%%%%%%%%PERIODIC CURVE
%%%%%%%%%%PERIODIC CURVE
\section{Curve of Periodic Points}\label{S:PeriodicCurve}

In this section, we prove

\begin{proposition}\label{P:CM_Transversal}
The central manifold of $P_1$ with respect to $T_a$ intersects $S_\lambda$ transversally for sufficiently small $\lambda>0$. 
\end{proposition}

\begin{remark}
The existence of a curve of periodic points is known from \cite{Can09}, see also \cite{HM07}. The construction in this section uses quadratic forms to obtain properties we use in the proof of Theorem~\reff{T:LinearGapIntro}.
\end{remark}

As both the central manifold and $S_\lambda$ are invariant under $T_a$, this transversal intersection gives us an invariant set of order two on $S_\lambda$ for sufficiently small $\lambda$. This shows that the central manifold is either a fixed curve or a curve of points of period two. To show this, we first take the $I(x,y,z)$ and shift $P_1$ to the origin. This gives us
$$J(x,y,z)=x^2 + y^2 + z^2 -2(xy+xz+yz)-2xyz.$$
Let us first consider the lower order terms of $J$,
$$J'(x,y,z)=x^2+y^2+z^2-2(xy+xz+yz).$$

\begin{remark}
We write the proofs in Sections~\reff{S:PeriodicCurve} and ~\reff{S:LinearGapOpening} for $D(T_a)$ for ease of notation. The results are readily extended to $D(\prod_{i=1}^nT_{a_i})(P_1)$ using the chain rule.
\end{remark}

\begin{lemma}\label{L:QF_Preserved}
$D(T_a)(P_1)$ preserves $J'$, i.e. for $\mathbf{v} \in \R^3$, $J'(\mathbf{v})=J'(D(T_a)(P_1)\mathbf{v})$.
\end{lemma}

\begin{proof}
It suffices to check this fact for $DU(P_1)$ and $D(UP)(P_1)$. Let $\mathbf{v}=(x,y,z)$.

$$\begin{array}{RCL}
&&J'(DU(P_1)\mathbf{v})\\&=& J'(2x-y+2z,x,z)\\
&=&(2x-y+2z)^2+x^2+z^2-2(x(2x-y+2z)+z(2x-y+2z)+xz)\\
&=&5x^2+y^2+5z^2-4xy+8xz-4yz-2(2x^2+2z^2-xy+5xz-yz)\\
&=&x^2+y^2+z^2-2(xy+xz+yz)\\
\\
&&J'(D(UP)(P_1)\mathbf{v})\\
&=& J'(2x+2y-z,x,y)\\
&=&(2x+2y-z)^2+x^2+y^2-2(x(2x+2y-z)+y(2x+2y-z)+xy)\\
&=&5x^2+5y^2+z^2+8xy-4xz-4yz-2(2x^2+2y^2+5xy-xz-yz)\\
&=&x^2+y^2+z^2-2(xy+xz+yz)
\end{array}$$
\end{proof}

We are interested in studying $q(x,y,z)=x^2+y^2-z^2$ since this can be obtained from $J'$ through a linear change of coordinates:
\begin{equation}\label{E:LinearCoC}
x^2+y^2-z^2=J'\left(x+\frac{3}{4}y+\frac{5}{4}z,-\frac{1}{4}y+\frac{1}{4}z,y+z\right).
\end{equation}

\begin{definition}\label{D:NullCone}
Let $q$ be an $n$-ary quadratic form over $\R$. We say that the \defword{null cone} of $q$ is the collection of $\mathbf{v} \in \R^n$ such that $q(\mathbf{v})=0$.
\end{definition}

\begin{proposition}\label{P:NullCone}
Let $q(x,y,z)=x^2+y^2-z^2$ and let $A$ be a $3 \times 3$ matrix that preserves $q$, i.e. for $v \in \R^3$, $q(A\mathbf{v})=q(\mathbf{v})$. Further suppose that $A$ has eigenvalues $|\mu|>1$, $|\nu|<1$ and $\pm1$. We say that $\mathbf{v}$ lies in the interior of the null cone if $q(\mathbf{v})<0$ and in the exterior of the null cone if $q(\mathbf{v})>0$. Any eigenvector corresponding to $\pm 1$, or central vector, lies in the exterior of the null cone of $q$.
\end{proposition}

To prove Proposition~\reff{P:NullCone}, we will first establish two preliminary lemmas. Let us first remind the reader of the polarization identity (See, for example \cite{Ger08}).

\begin{lemma}\label{L:QF_Norm}
Let $q$ be an $n$-ary quadratic form over $\R$. We may associate to $q$ a symmetric bilinear form $B$ on $\R^n$ such that for $\mathbf{v} \in \R^n$, $q(\mathbf{v})=B(\mathbf{v},\mathbf{v})$. Then, for $\mathbf{v},\mathbf{w} \in \R^n$, we have %the following identity, sometimes called the \defword{polarization identity}, holds
$$B(\mathbf{v},\mathbf{w})=\frac{B(\mathbf{v}+\mathbf{w},\mathbf{v}+\mathbf{w})-B(\mathbf{v},\mathbf{v})-B(\mathbf{w},\mathbf{w})}{2}.$$
\end{lemma}

\begin{lemma}\label{L:StabUnstabNullCone}
Suppose $q$ is a ternary quadratic form and $A$ is a $3 \times 3$ matrix that preserves $q$. Further suppose that $A$ has eigenvalues $|\mu|>1$, $|\nu|<1$ and $\pm1$. Any eigenvector corresponding to $\mu$, or unstable vector, (resp. $\nu$, or stable vector) lies on the null cone.
\end{lemma}

\begin{proof}
Let $\mathbf{v}$ be an eigenvector corresponding to $\mu$. Then
$$B(\mathbf{v},\mathbf{v})=B(\mu\mathbf{v},\mu\mathbf{v})=\mu^2 \cdot B(\mathbf{v},\mathbf{v}) \Rightarrow q(\mathbf{v})=B(\mathbf{v},\mathbf{v})=0.$$
Repeating the argument for $\nu$ in place of $\mu$ shows that eigenvectors corresponding to $\nu$ must also lie on the null cone.
\end{proof}

\begin{proof}[Proof of Proposition \reff{P:NullCone}.]
Suppose that $\mathbf{v}$ and $\mathbf{w}$ are unstable and central vectors of $A$, respectively, then
$$B(\mathbf{v},\mathbf{w})=B(A\mathbf{v},A\mathbf{w})=B(\mu \mathbf{v},\pm\mathbf{w})=\pm \mu \cdot B(\mathbf{v},\mathbf{w}) \Rightarrow B(\mathbf{w},\mathbf{v})=0.$$
Since we established in Lemma~\reff{L:StabUnstabNullCone} that $\mathbf{v}$ lies in the null cone of $q$, by scaling and a rotation, we may assume that $\mathbf{v}=(1,0,1)$. If we denote $\mathbf{w}=(w_1,w_2,w_3)$, then $B(\mathbf{v},\mathbf{w})=w_1-w_3=0$. If $w_2=0$, then $\mathbf{w}$ lies in the same eigenspace as $\mathbf{v}$, but this is a contradiction as $\mathbf{w}$ is a central vector. Hence $w_2 \neq 0$ and $q(\mathbf{w})=(w_2)^2 > 0$. Thus $\mathbf{w}$ lies in the exterior of the null cone.
\end{proof}

\begin{lemma}\label{L:TransversalPertubation}
Suppose that $q(x,y,z)=x^2+y^2-z^2$ and $p(x,y,z)$ is a homogeneous polynomial of degree 3. Suppose that $\ell$ is a line passing through the origin and is transversal to the level surface $$\calS_\lambda:=\set{(x,y,z):q(x,y,z)=\lambda^2}$$ for all sufficiently small $\lambda>0$. Then $\ell$ is transversal to the level surface $$\tilde{\calS}_\lambda:=\set{(x,y,z):q(x,y,z)+p(x,y,z)=\lambda^2}$$ for all sufficiently small $\lambda>0$. Further suppose that $\tilde{\ell}$ be a smooth curve passing through the origin tangent to $\ell$. Then $\tilde{\ell}$ is transversal to the level surface $\tilde{\calS}_\lambda$ for sufficiently small $\lambda>0$.
\end{lemma}

\begin{proof}
%Considering the problem for sufficiently small $\lambda$ is equivalent to considering it in an open ball of sufficiently small radius $r>0$ about the origin. Since $p$ is a homogeneous polynomial of degree 3, $|\partial p/\partial x|$, $|\partial p/\partial y|$, $|\partial p/\partial z| \sim r^2$ while $\nabla q = (2x, 2y, -2z)$. Thus for sufficiently small $r$, $\nabla(q+p)$ is never identically zero.

Suppose $\ell$ is the span of a vector $\mathbf{v}$. By hypothesis, $\mathbf{v}$ is transversal to $\calS_c$ and hence $\nabla q \cdot \mathbf{v} \neq 0$. In the perturbed case, we wish to show that $\nabla(q+p) \cdot \mathbf{v} \neq 0$. Since $p$ is a homogeneous polynomial of degree 3, $|\partial p/\partial x|$, $|\partial p/\partial y|$, $|\partial p/\partial z| \sim r^2$ while $\nabla q = (2x, 2y, -2z)$. Thus, in a ball about the origin of sufficiently small radius $r>0$, $|\nabla q \cdot \mathbf{v}| > |\nabla p \cdot \mathbf{v}|$. This gives us our desired result.

Now consider the perturbation of $\ell$ to $\tilde{\ell}$. Suppose that $\tilde{\mathbf{v}}$ is tangent to $\tilde{\ell}$ at an intersection with $\tilde{\calS}_\lambda$. Since $\tilde{\ell}$ is smooth, $\lambda$ can be chosen so that $\tilde{\mathbf{v}}$ is arbitrarily close to $\mathbf{v}$ and thus $\nabla(q+p) \cdot \tilde{\mathbf{v}} \neq 0$.
\end{proof}

\begin{proof}[Proof of Proposition~\reff{P:CM_Transversal}]
Observe that Proposition~\reff{P:NullCone} justifies the hypotheses of Lemma~\reff{L:TransversalPertubation}. That is, the tangent vector to the central manifold lies in the exterior or the null cone, which is the level surface given by $J'(x,y,z)=0$ (modulo the change of coordinates \eqref{E:LinearCoC}). Applying Lemma~\reff{L:TransversalPertubation} to the level surface given by $J(x,y,z)$ (modulo the change of coordinates \eqref{E:LinearCoC}) gives us Proposition~\reff{P:CM_Transversal}.
\end{proof}

%%%%%%%%%%LinearGapOpening
%%%%%%%%%%LinearGapOpening
%%%%%%%%%%LinearGapOpening
%%%%%%%%%%LinearGapOpening
%%%%%%%%%%LinearGapOpening
\section{Linear Gap Opening}\label{S:LinearGapOpening}

In this section, we prove

\begin{thm:linear}
For $\alpha$ with eventually periodic continued fraction expansion and $\lambda>0$ sufficiently small, the boundary points of a gap in the spectrum $\Sigma_{\lambda,\alpha}$ depend smoothly on the coupling constant $\lambda$. Moreover, given any one-parameter continuous family $\set{U_{\lambda,\alpha}}_{\lambda>0}$ of gaps of $\Sigma_{\lambda,\alpha}$, we have that
$$\lim_{\lambda \to 0^+} \frac{|U_{\lambda,\alpha}|}{|\lambda|}$$
exists and belongs to $(0,\infty)$.
\end{thm:linear}

To show this, we use an invariant set of order 2 which we found in Section~\reff{S:PeriodicCurve} to moderate the growth of gap length $|U_\lambda|$ via Lemma~\reff{L3.2:DG11}.

\begin{lemma}\label{L:LinearPertubation}
Retaining the assumptions of Lemma \reff{L:TransversalPertubation}, further suppose that $\tilde{\ell}$ be a smooth curve passing through the origin tangent to $\ell$ and intersecting $\calS_\lambda$ at two points for sufficiently small $\lambda$. %, call them $\mathbf{p}_\lambda$ and $\mathbf{q}_\lambda$.
Then there is a neighborhood of the origin $U$ such that $\tilde{\calS}_\lambda \cap \tilde{\ell} \cap U = \set{\tilde{\mathbf{p}}_\lambda, \tilde{\mathbf{q}}_\lambda}$ and $d(\tilde{\mathbf{p}}_\lambda, \tilde{\mathbf{q}}_\lambda) \sim \lambda$, that is
$$\lim_{\lambda \to 0^+}\frac{d(\tilde{\mathbf{p}}_\lambda, \tilde{\mathbf{q}}_\lambda)}{\lambda}$$ exists and belongs to $(0,\infty)$.
\end{lemma}

\begin{proof}
A direct computation shows that this is true in the case that $\tilde{\ell}$ is a line passing through the origin and $p(x,y,z)\equiv0$, i.e. we are looking at the unperturbed surface $\calS_c$. Thus the case that $\tilde{\ell}$ is a smooth curve is a perturbation of this case and looking at the Taylor expansion of $\tilde{\ell}$ shows that we add an error term of order at most $\lambda^2$. 
A direct application of the Morse lemma (see for example \cite{Mat02}) shows that there is a neighborhood of the origin $U$ and a smooth change of coordinates $H: U \to \R^3$ such that $H(0,0,0)=(0,0,0)$ and $p \circ H \inv (x,y,z)=x^2+y^2-z^2$, then our desired result would follow. 
\end{proof}

\begin{lemma}[Lemma 3.2 in \cite{DG11}]\label{L3.2:DG11}
Let $W\subset \mathbb{R}^3$ be a smooth surface with a $C^1$-foliation on it. Let $\xi\subset W$ be a smooth curve transversal to the foliation. Fix a leaf $L\subset W$ of the foliation, and denote $P=L\cap \xi$. Take a point $Q\ne P$, $Q\in L$, and a line $\ell_0\subset \mathbb{R}^3$, $Q\in \ell_0$, tangent to $W$ at $Q$, but not tangent to the leaf $L$. Suppose that a family of lines $\{\ell_\lambda\}_{\lambda\in (0, \lambda_0)}$ is given such that $\ell_\lambda \to \ell_0$ as $\lambda \to 0$, each line $\ell_\lambda$, $\lambda>0$,  intersects $W$ at two points $p_\lambda$ and $q_\lambda$, and $p_\lambda\to Q$, $q_\lambda\to Q$ as $\lambda\to 0$. 
Denote by $L_{p_\lambda}$ and $L_{q_\lambda}$ the leaves of the foliation that contain $p_\lambda$ and $q_\lambda$, respectively. Denote ${\bf{p}}_\lambda=L_{p_\lambda}\cap \xi$ and ${\bf{q}}_\lambda=L_{q_\lambda}\cap \xi$. Then there exists a finite non-zero limit
 $$
 \lim_{\lambda\to 0}\frac{\mathrm{dist}(p_\lambda, q_\lambda)}{\mathrm{dist}({\bf{p}}_\lambda, {\bf{q}}_\lambda)}.
 $$
 \end{lemma}
 
\begin{proof}[Proof of Theorem~\reff{T:LinearGapIntro}]
Proposition~\reff{P:NullCone} shows that the central manifold of $T_a$ at the point $(1,1,1)$ is transversal to $S_\lambda$ for sufficiently small $\lambda>0$ in a neighborhood of $P_1$. Since the central manifold is invariant and $S_\lambda$ is invariant and they intersect at exactly two points, these two points form an invariant set of size 2, that is two fixed points or a set of period 2. Call these two points $p_1(\lambda)$ and $p_1'(\lambda)$. One may verify by direct computation that the maps $T_a$ fix $P_1$. Also, if $a$ is even, then $T_a$ fixes $P_4$ and permutes $P_2$ and $P_3$. If $a$ is odd, $P_2$, $P_3$, $P_4$ form a 3-cycle under $T_a$. Thus, $T_a^6$ fixes all $P_i$ for $i=1...4$. Thus $p_1(\lambda)$ and $p_1'(\lambda)$ will be two fixed points of $T_a^6$. Thus, in a neighborhood of radius $r_0$ of $P_1$, we have a smooth curve of fixed points, call this curve $\Fix(T_a^6, O_{r_0}(P_1))$. Due to Lemmas \reff{L:TransversalPertubation} and \reff{L:LinearPertubation} the length of the curve between the two points of intersection grows linearly with respect to $\lambda$. Applying Lemma \reff{L3.2:DG11} to $W^s(\Fix(T_a^6, O_{r_0}(P_1))$ gives us Theorem \reff{T:LinearGapIntro}.
\end{proof}

%%%%%%%%%%ContHausDim
%%%%%%%%%%ContHausDim
%%%%%%%%%%ContHausDim
%%%%%%%%%%ContHausDim
%%%%%%%%%%ContHausDim
\section{Continuity of the Hausdorff Dimension}\label{S:ContHausDim}
In this section, we prove
\begin{thm:hausdim}
For $\alpha$ with eventually periodic continued fraction expansion,
$$\lim_{\lambda \to 0^+} \dim_\text{H} \Sigma_{\lambda, \alpha} = 1.$$
More precisely, there are constants $C_1, C_2 > 0$ such that
$$1-C_1 \lambda \leq \dim_\text{H} \Sigma_{\lambda, \alpha} \leq 1-C_2 \lambda$$
for $\lambda>0$ sufficiently small.
\end{thm:hausdim}
Let us begin by recalling the notions of \defword{thickness} $\tau$ and \defword{denseness} $\theta$ (see, for example, \cite{PT93}).

\begin{definition}\label{D:Thickness}
Let $C \subset \R$ be a Cantor set and $I$ be the minimal closed interval of $\R$ containing $C$. A bounded \defword{gap} of $C$ is a bounded connect component of $\R \bs C$. An enumeration of these bounded gaps $\calU=\set{U_n}$ is called a \textit{presentation}. If $u \in C$ is a boundary point of $U_n$, denote by $K$ the connected component of $I \bs (U_1 \cup ... \cup U_n)$ containing $u$. We say $K$ is a \defword{bridge}. We define
$$\tau(C,\calU,u):=\frac{|K|}{|U_n|}$$
and call this the \defword{thickness of $C$ at $u$}.
Further,
$$\tau(C)=\sup_{\calU} \inf_{u} \tau(C,\calU,u), \qquad \theta(C)=\inf_{\calU} \sup_{u} \tau(C,\calU,u).$$
\end{definition}

\begin{proposition}[Proposition 5 and 6 in Section 4.2 of \cite{PT93}]\label{P:ThicknessHausDim}
Thickness and denseness are related to Hausdorff dimension by the following inequalities:
$$\frac{\log 2}{\log \left(2+\frac{1}{\tau(C)}\right)} \leq \dim_\text{H} C  \leq \frac{\log 2}{\log \left(2 + \frac{1}{\theta(C)}\right)}.$$
\end{proposition}

Observe that Proposition~\reff{P:ThicknessHausDim} and the following theorem immediately imply Theorem~\reff{T:HausDimIntro}.

\begin{theorem}\label{T:Thickness}
For $\alpha$ with eventually periodic continued fraction expansion,
$$\lim_{\lambda \to 0} \tau(\Sigma_{\lambda, \alpha}) = \infty.$$
More precisely, there are constants $C_3, C_4 > 0$ such that
$$C_3 \lambda \inv \leq \tau(\Sigma_{\lambda, \alpha}) \leq \theta(\Sigma_{\lambda, \alpha}) \leq C_4 \lambda \inv$$
for $\lambda>0$ sufficiently small.
\end{theorem}

\begin{proof}
The proof of Theorem \reff{T:Thickness} is almost identical to the proof of Theorem 1.1 in \cite{DG11}, distortion estimates can be obtained subject to two observations:
\begin{enumerate}
\item Due to the smoothness of the invariant manifolds of the periodic curve (see Section~\reff{S:PeriodicCurve}), there exists a smooth change of coordinates in a neighborhood $U$ of $P_1$. $\Phi:U(P_1)\to \mathbb{R}^3$, $\Phi(P_1)=(0,0,0)$ that rectifies all invariant manifolds. The exact form is given in Lemma~\reff{L:Rectification}.
%Denote by $I$ the set $\Fix(T^6, O_{r_0}(P_1)$ we studied in the Section~\reff{S:LinearGapOpening} and notice that it is normally hyperbolic with respect to $T_\alpha^\text{per}$. %(for precise definitions and details, see \cite{PT93}).
%Due to the smoothness of the invariant manifolds of the curve of fixed points, there exists a smooth change of coordinates $\Phi:O_{r_0}(P_1)\to \mathbb{R}^3$, $\Phi(P_1)=(0,0,0)$ that rectifies all invariant manifolds
\item As can be seen from the Definition \reff{D:Thickness}, since $\tau(\Sigma_{\lambda, \alpha})$ is the supremum over all presentations $\calU$ showing that thickness relative to a particular presentation tends to infinity is sufficient for the proof of Theorem~\reff{T:Thickness}. To apply Lemma 3.10 in \cite{DG11}, we must construct a partition comprised of rectangles whose boundaries are stable and unstable manifolds. The existence of such a partition is shown in Lemma~\reff{L:Partition} and Remark~\reff{R:Partition}.
\end{enumerate}
\end{proof}

\begin{remark}\label{R:Partition}
In Section \reff{S:HyperbolicityNW}, we defined $W^S(\Omega_{\lambda, \alpha})$ to be the preimage of the stable lamination with respect to $T_\alpha^\text{per}$ under $T_\alpha^\text{int}$. Similarly, we will refer to the preimage of stable (resp. unstable) manifolds with respect to $T_\alpha^\text{per}$ on $S_\lambda$ under $T_\alpha^\text{int}$ as the stable (resp. unstable) manifolds of $T_\alpha$. 
\end{remark}

\begin{lemma}\label{L:Partition}
For $\alpha$ with eventually periodic continued fraction expansion, there is a Markov partition on $\Omega_{\lambda, \alpha}^\text{per} \subseteq S_\lambda$ for $\lambda$ sufficiently small. Using the notation of Remark~\reff{R:Partition}, the preimage of this Markov partition under $T_\alpha^\text{int}$ gives us a partition of $\Omega_{\lambda, \alpha}$ composed of pieces of stable and unstable manifolds of $T_\alpha$. 
\end{lemma}

\begin{proof}
As we remarked in Section \reff{S:HyperbolicityNW}, in the periodic case $\alpha=[0;\overline{a_{1}...a_{n}}]$, and $T_\alpha^{\text{per}}=T_{a_{n}} \circ ... \circ T_{a_{1}}$. Then $T_\alpha^{\text{per}}$ is a factor of $M_{a_n} \circ ... \circ M_{a_1}$ under the semi-conjugacy $F$, which is a hyperbolic toral automorphism. Thus, a  Markov partition for this map can be constructed on $\T^2:=\R^2/\Z^2$ from the stable and unstable manifolds of the points $\tilde{P}_1=(0,0), \tilde{P}_2=(1/2,0), \tilde{P}_3=(0,1/2), \tilde{P}_4=(1/2,1/2)$. This is a classical result, see, for example, Appendix 2 in \cite{PT93}. We may label this construction in such a way that it is invariant under the map $(x,y) \mapsto (-x,-y)$. Under the semi-conjugacy $F$, this induces a Markov partition on $\Omega_{0, \alpha}^\text{per} \subset \S$. Since this partition is composed of piece of stable and unstable manifolds of $T_\alpha^\text{per}$, which depend smoothly on $\lambda$, there exists $\lambda_0>0$ such that the partition can be extended to $\Omega_{\lambda, \alpha}^\text{per} \subseteq S_\lambda$ for $\lambda \in [0, \lambda_0)$.%We will also refer to the preimage of stable and unstable manifolds of $M_\alpha^\text{per}$ under $M_\alpha^\text{int}$ on $\T^2$ as the stable and unstable manifolds of $M_\alpha$.
\end{proof}

In order to present the exact form of the rectification $\Phi$, we make use of some standard notations associated with normal hyperbolicity (see \cite{PT93} for a formal definition). From Section~\reff{S:PeriodicCurve}, there is a smooth curve that was normally hyperbolic with respect to $T_\alpha^\text{per}$ (the existence of this curve is shown in Section~\reff{S:LinearGapOpening}) and denote its intersection with a neighborhood $U$ of $P_1$ by $I$. Define the local center-stable and center-unstable manifolds
\begin{equation*}
W^{cs}_{loc}(I)=(T_\alpha^\text{int})\inv\set{p \in U: (T_\alpha^\text{per})^n (p) \in U \text{ for all } n \in \N},
\end{equation*}
\begin{equation*}
W^{cu}_{loc}(I)=(T_\alpha^\text{int})\set{p \in U: (T_\alpha^\text{per})^{-n} (p) \in U \text{ for all } n \in \N}.
\end{equation*}
Define also local strong-stable and strong-unstable manifolds
\begin{equation*}
W^{ss}_{loc}(P_1)=(T_\alpha^\text{int})\inv\set{p \in (T_\alpha^\text{int})W^{cs}_{loc}(I): (T_\alpha^\text{per})^n(p) \to P_1 \text{ as } n \to + \infty}
\end{equation*}
\begin{equation*}
W^{uu}_{loc}(P_1)=(T_\alpha^\text{int})\set{p \in (T_\alpha^\text{int})\inv W^{cu}_{loc}(I): (T_\alpha^\text{per})^{-n}(p) \to P_1 \text{ as } n \to + \infty}
\end{equation*}

\begin{lemma}\label{L:Rectification}
Then there exists a smooth change of coordinates $\Phi:U(P_1)\to \mathbb{R}^3$ such that $\Phi(P_1)=(0,0,0)$ and
\begin{itemize}
\item
$\Phi(I)$ is a part of the line $\{x=0, z=0\}$;
\item
$\Phi(W^{cs}_{loc}(I))$ is a part of the plane $\{z=0\}$;
\item
$\Phi(W^{cu}_{loc}(I))$ is a part of the plane $\{x=0\}$;
\item
$\Phi(W^{ss}_{loc}(P_1))$ is a part of the line $\{y=0, z=0\}$;
\item
$\Phi(W^{uu}_{loc}(P_1))$ is a part of the line $\{x=0, y=0\}$.
\end{itemize}
\end{lemma}

\begin{proof}
The existence of a rectification $\Phi$ follows from the smoothness of all invariant manifolds involved, see for example \cite{HPS77}.
\end{proof}

\begin{corollary}\label{C:TransInt_InitialCond}
For $\alpha$ with eventually periodic continued fraction expansion, there exists $\lambda_0>0$ such that for $\lambda \in (0,\lambda_0)$, for every $\mathbf{v} \in \Omega_{\lambda, \alpha}$, the stable manifold $W^s(\mathbf{v})$ intersects the line $\ell_\lambda$ transversally.
\end{corollary}

\begin{proof}
In the case of the Fibonacci Hamiltonian, this is Theorem 3 (iii) in \cite{DG09a}. Given the rectification from Lemma~\reff{L:Rectification}, this follow from directly from Lemma 5.5 in \cite{DG09a}
\end{proof}

%%%%%%%%%%IDOS
%%%%%%%%%%IDOS
%%%%%%%%%%IDOS
%%%%%%%%%%IDOS
%%%%%%%%%%IDOS
\section{Integrated Density of States}\label{S:IDOS}
Restarting the definition for convenience, let $H^\mathfrak{I}_{\lambda, \alpha}$ be $H_{\lambda, \alpha}$ restricted to a finite interval $\mathfrak{I}$ with Dirichlet boundary conditions and let the \defword{integrated density of states} be given by
\begin{equation*}\tag{E:IDOS}
N_{\lambda, \alpha}(E)=\lim_{L \to \infty} \frac{1}{L} \left|\left\{\text{eigenvalues of }H_{\lambda, \alpha}^{[1,L]} \text{ that lie in }(-\infty,E)\right\}\right|.
\end{equation*}
See \cite{Hof93} for a proof that this limit exists and is independent of $\omega \in \T$. See also \cite{KM07} for a survey of the existence of this limit in a very general setting as well as the following equivalent definition. 

The density of states can also be obtained by spectral considerations. Recall that by the spectral theorem, there are Borel probability measures $d\mu_{\lambda, \alpha, \omega}$ on $\R$ such that for bounded, measurable $g$ and Dirac delta $\delta_0$,
$$\brac{\delta_0, g(H_{\lambda, \alpha, \omega})\delta_0}=\int g(E) d\mu_{\lambda, \alpha, \omega}.$$
Averaging this over the phase $\omega$ (with respect to Lebesgue measure)
$$\int_{\omega \in \T} \brac{\delta_0, g(H_{\lambda, \alpha, \omega})\delta_0}=\int g(E) dN_{\lambda, \alpha}(E)$$
gives us the density of states measure $dN_{\lambda, \alpha}(E)$. %By this measure is non-atomic and its support is $\Sigma_{\lambda, \alpha}$.
%The existence of this limit has been shown in a general setting in \cite{LS05}. It has also been studied for general potentials, random potentials, and analytic quasi-periodic potentials, see Section 5.5 of \cite{DEG13} for references. This will be developed more formally in Section~\reff{S:IDOS}, but for now we simply state that $N_{\lambda, \alpha}$ is the cumulative distribution function of the spectral measure supported on $\Sigma_{\lambda, \alpha}$, which we denote $dN_{\lambda, \alpha}$ and call the \defword{density of states measure}. 
Note that zero coupling constant means that there is no potential and hence $\alpha$ does not affect these values. For convenience, let us introduce the notation $\Sigma_0:=\Sigma_{0, \alpha}$ and $N_{0,\alpha}:=N_0$. It is a standard result that $\Sigma_{0} = [-2,2]$ and
$$N_{0}(E)=\begin{cases} 0 & E \leq - 2 \\ \frac{1}{\pi}\arccos(-\frac{E}{2}) & -2 <E<2\\ 1 & E\geq 2  \end{cases}.$$

\begin{thm:exactdim}\label{T:ExactDimBody}
For $\alpha$ with eventually periodic continued fraction expansion, there exists $0<\lambda_0\leq \infty$ such that for $\lambda \in (0, \lambda_0)$, there is $d_{\lambda, \alpha} \in (0,1)$ so that $dN_{\lambda, \alpha}$ is of exact dimension $d_{\lambda, \alpha}$, that is, for $dN_{\lambda, \alpha}$-almost every $E \in \R$,
$$\lim_{\epsilon \to 0^+} \frac{\log(N_{\lambda, \alpha}(E+\epsilon)-N_{\lambda,\alpha}(E-\epsilon))}{\log \epsilon}=d_{\lambda, \alpha}.$$
Moreover, in $(0, \lambda_0)$, $d_{\lambda, \alpha}$ is a $C^\omega$ function of $\lambda$ and $$\lim_{\lambda \to 0^+} d_{\lambda, \alpha}=1.$$
\end{thm:exactdim}

See also \cite{Gir13} for related results.

\begin{proposition}[Claim 3.1 in \cite{DG12}]\label{P:PushForward}
%Let $\alpha=[0;a_1,...,a_m,\overline{a_{m+1}...a_{m+n}}]$ and 
Let $I=[0, 1/2] \times \set{0}$. The push forward of the normalized Lebesgue measure on $I$ under the semi-conjugacy $F$, which is a probability measure on $\ell_0 \cap \S$, corresponds to the free density of states measure under the identification
$$J_\lambda: E \mapsto \left(\frac{E-\lambda}{2}, \frac{E}{2}, 1 \right)$$
\end{proposition}

%\begin{proof}[Proof of Proposition \reff{P:PushForward}] Let us first consider the case that that $m=0$, i.e. the continued fraction expansion of $\alpha$ is periodic. This follows exactly as in Claim 3.1 in \cite{DG12}, with the observation that $M_\alpha^\text{per}$ is a factor of a toral hyperbolic automorphism.

%Now let us return to the more general case. The measure on this interval is obtained by projection of normalized Lebesgue measure on $I$. In this case, we create a rectangle using the stable and unstable manifolds of $M_\alpha$ (Recall the partition from Lemma~\reff{L:Partition} and Remark~\reff{R:Partition}.) emanating from $(0,0)$ and $(0,1/2)$. Call this rectangle $R$. Projecting the normalized Lebesgue measure on $R$ onto the interval $I=[0, 1/2] \times \set{0}$ allows us to proceed as in the periodic case.
%\end{proof}

\begin{proof}[Proof of Theorem \reff{T:ExactDimIntro}]
Recall from Lemma~\reff{L:Partition} and Remark~\reff{R:Partition} that there exists a partition on $\Omega_{\lambda, \alpha}$ consisting of rectangles whose boundaries are stable and unstable manifolds. We construct family of measures $\set{\nu_{\lambda, \alpha}}_{\lambda \in [0, \lambda_0)}$ by projecting the normalized Lebesgue measure on rectangle containing $\ell_\lambda$ onto $\ell_\lambda$ along the stable lamination.

For a fixed $\alpha$ with eventually periodic continued fraction expansion, we claim $\set{\nu_{\lambda, \alpha}}_{\lambda \in [0, \lambda_0)}$ satisfies the following properties:
\begin{enumerate}
\item $\supp (\nu_{\lambda, \alpha})=\Sigma_{\lambda, \alpha}$
\item $\nu_{0, \alpha}=dN_0$
\item $\set{\nu_{\lambda, \alpha}}$ depends continuous on $\lambda$ for $\lambda \in [0, \lambda_0)$
\item for any two continuous families of gaps $\set{U_{\lambda, \alpha}}_{\lambda \in [0, \lambda_0)}$ and $\set{W_{\lambda, \alpha}}_{\lambda \in [0, \lambda_0)}$ in the spectrum $\Sigma_{\lambda, \alpha}$, the measure $\nu_{\lambda, \alpha}(E_1, E_2)$, for $E_1 \in U_{\lambda, \alpha}$ and $E_2 \in W_{\lambda, \alpha}$, is independent of $\lambda$.
\end{enumerate}
Part (1) follows from Lemma~\reff{L:InitialConditions} and Part (2) follows from Proposition~\reff{P:PushForward}. Part (3) follows from the continuous dependence of $\Omega_{\lambda, \alpha}$ on $\lambda$. Recall from Remark~\reff{R:StableLamination} that gaps open at precisely the energies in $\Sigma_0$ corresponding intersections of $\ell_0$ with the stable lamination $W^S(\Omega_{\lambda, \alpha})$. It also follows from Lemma~\reff{L:Partition} and Remark~\reff{R:Partition} that the partition on $\Omega_{\lambda, \alpha}$ that depends continuous on $\lambda$. However, by this continuity, the rectangle of the partition in which the gap lies is $\lambda$-independent for sufficiently small $\lambda$. This gives us Part (4). Finally, from this, the proof follows as in the proof of Theorem 1.1 in \cite{DG12}. \end{proof}

\begin{lemma}\label{L:ExpandingDirection}
Let $\alpha=[0;a_1...a_n\overline{a_{n+1}...a_{n+k}}]$, that is $\alpha$ is a rotation angle with eventually periodic continued fraction expansion.  Then $\displaystyle \lim_{j \to \infty} \mathbf{v}_j$ has slope $\alpha \inv$, where $$\mathbf{v}_{j}:=\left(M_{a_j} \circ ... \circ M_{a_1}\right)^T \begin{pmatrix}1\\0\end{pmatrix}.$$\end{lemma}

\begin{proof}
Recall that $p_n/q_n=[0;a_1,...,a_n]$ is the $n$th rational approximant to $\alpha$. It then follows from induction and Equations \reff{E:RationalApproximantP} and \reff{E:RationalApproximantQ} that $$M_{a_j} \circ ... \circ M_{a_1} = \begin{pmatrix} q_j & p_j \\ q_{j-1} & p_{j-1} \end{pmatrix}.$$
This gives us our desired result.
\end{proof}

\begin{thm:lessthandim}
For $\alpha$ with eventually periodic continued fraction expansion and for $\lambda>0$ sufficiently small, we have $$d_{\lambda, \alpha}<\dim_\text{H} \Sigma_{\lambda, \alpha}.$$
\end{thm:lessthandim}

\begin{proof}
For equality of two quantities in Theorem \reff{T:LessThanDimIntro}, the average multipliers over all periodic points of $T_\alpha|_{\Omega_{\lambda,\alpha}}$ must be equal (see the proof of Theorem 1.2 in \cite{DG12} for a detailed explanation of this claim).

Recall Lemma~\reff{L:Semi-conjugacy}, that the map on the surface $\S$ is a factor of a linear map on $\T^2$ via the semi-conjugacy $F: (\theta,\phi) \mapsto (\cos(2 \pi (\theta + \phi)), \cos(2\pi\theta), \cos(2\pi\phi))$. We claim that the expansion rate at the singularity $(1,1,1)$ is squared that of any periodic point which is not a singularity. Since the map on $\T^2$ is linear, the expansion rate on $\T^2$ is constant. Thus, at any periodic point where $DF$ is invertible, the expansion rate at a point $(\theta,\phi)$ on $\T^2$ and $F(\theta,\phi)$ on $\S$ are equal. Now let us consider the eigenline $L(t)=(t, \alpha t)$ (see Lemma~\reff{L:ExpandingDirection}) corresponding to the expanding direction on $\T^2$ under $F$. Observe that $DF$ is not invertible at $(0,0)$ and recall that the first two terms of the Taylor expansion of $\cos(x)$ are $1-\frac{x^2}{2}$. If we consider $\tilde{F}:=F-(1,1,1)$ so that $\tilde{F}$ maps $(0,0)$ to $(0,0,0)$ and $\tilde{F}$ is of order $x^2$ in each coordinate. Consider the image of $L(t)$ under $\tilde{F}$. Direct computation shows that the expansion rate at $(0,0,0)$ is squared that of the constant expansion rate on $\T^2$. This proves our claim and completes the proof of the theorem.
\end{proof}

We also prove that all gaps allowed by the gap labeling theorem are open.

\begin{thm:gaplabel}
For $\alpha$ with eventually periodic continued fraction expansion, there is $\lambda_0>0$ such that for every $\lambda \in (0, \lambda_0]$, all gaps of $H_{\lambda, \alpha}$ allowed by the gap labeling theorem are open. That is $$\set{N_{\lambda,\alpha}(E): E \in \R \bs \Sigma_\lambda}=\set{{m\alpha} : m \in \Z} \cup \set{1}.$$
\end{thm:gaplabel}

\begin{proof}[Proof of Theorem \reff{T:GapLabelingIntro}]
First, for simplicity, let $a \in \N$ and let $\alpha=[0;\overline{a}]$, that is $\alpha\inv$ is a metallic mean. Recall that $\tilde{P}_i$ for $i=1...4$, the preimages under $F$ of the singularities $P_i$, are periodic points of $M_a$ and that $T_\alpha$ is a factor of $M_a$. Recall also that $M_a$ is a $2 \times 2$ matrix acting on $\T^2$ with eigenvalues $\alpha$ and $\alpha \inv$, this follows from an examination of the characteristic polynomial of $M_a$. The stable manifolds of these points intersect the pre-image of $\ell_\lambda$ transversally at the points $\set{m \alpha}$, $\set{m \alpha + \frac{1}{2}}$, $\set{m \alpha + \frac{\alpha}{2}}$, $\set{m \alpha + \frac{1}{2} +\frac{\alpha}{2}}$ for $k \in \Z$.

The images of these points under $F$ are points on the line $\ell_0$ of the form $(\pm \cos(\pi m \alpha), \pm \cos(\pi m \alpha), 1)$. The integrated density of states for the free Laplacian takes values $\set{m \alpha}$ at the point $\set{m \alpha}$. Once we increase $\lambda$ so that it is nonzero, each point of the form $\set{m \alpha}$ gives rise to a pair of stable manifolds. Every point which is between the stable manifolds has an unbounded orbit and thus does not belong to $\Omega_{\lambda,\alpha}$. Hence, for $m \in \Z$, $\set{m \alpha}$ gives rise to an interval outside of $\Sigma_{\lambda, \alpha}$. Since $N_{\lambda, \alpha}(-\infty, E)$ depends continuously on $\lambda$ and since $N_{\lambda, \alpha}(-\infty, \cdot)$ is constant on the complement of $\Sigma_{\alpha, \lambda}$, the integrated density of states takes the same value in the gap as at the energy corresponding to the initial point of intersection of the stable manifold of the singularity with $\ell_0$.

In the general case that $\alpha=[0;a_1...a_n\overline{a_{n+1}...a_{n+k}}]$, Lemma \reff{L:ExpandingDirection} tells us that $\alpha$ is the slope of the stable direction of $T_\alpha$. The proof then follows as above.
\end{proof}

\bibliography{biblio}
\bibliographystyle{alpha}
\end{document}